\documentclass[a4paper,fleqn]{cas-dc}

\usepackage[numbers,sort&compress]{natbib}
\usepackage{amsmath, amssymb}
\usepackage[linesnumbered, ruled, vlined]{algorithm2e}
\usepackage{graphicx}
\usepackage{subcaption}
\usepackage{placeins}  % 或不带[section]
\usepackage{xcolor}
\usepackage{booktabs}
\usepackage{tikz}
\usepackage{etoolbox}
\usepackage{xspace}
\usetikzlibrary{positioning,calc}

\def\tsc#1{\csdef{#1}{\textsc{\lowercase{#1}}\xspace}}
\tsc{WGM}
\tsc{QE}
\tsc{EP}
\tsc{PMS}
\tsc{BEC}
\tsc{DE}

\newtheorem{assumption}{Assumption}
\newtheorem{definition}{Definition}
\newtheorem{proposition}{Proposition}

\newtheorem{theorem}{Theorem}
\newtheorem{lemma}{Lemma}
\newdefinition{remark}{Remark}
\newproof{pf}{Proof}
\newproof{pot}{Proof of Theorem \ref{thm}}

\begin{document}
\let\WriteBookmarks\relax
\def\floatpagepagefraction{1}
\def\textpagefraction{.001}

\shorttitle{Data-driven online control for real-time optimal economic dispatch and temperature regulation in district heating systems}

% Short author
\shortauthors{Xinyi Yi and Ioannis Lestas}

\title[mode=title]{Data-driven online control for real-time optimal economic dispatch and temperature regulation in district heating systems}

\author[1]{Xinyi Yi}[orcid=0000-0003-1797-6280]
\ead{xy343@cam.ac.uk}

\author[1]{Ioannis Lestas}
\ead{icl20@cam.ac.uk}
\cormark[1]

\cortext[1]{Corresponding author.}

% Address/affiliation
\affiliation {organization={Department of Engineering, University of Cambridge},
    addressline={Trumpington Street}, 
    city={Cambridge},
    % citysep={}, % Uncomment if no comma needed between city and postcode
    postcode={CB2 1PZ}, 
    % state={},
    country={United Kingdom}}

\begin{abstract}
District heating systems (DHSs) require coordinated economic dispatch and temperature regulation under uncertain operating conditions. Existing DHS operation strategies often rely on disturbance forecasts and nominal models, so their economic and thermal performance may degrade when predictive information or model knowledge is inaccurate. This paper develops a data-driven online control framework for DHS operation by embedding steady-state economic optimality conditions into the temperature dynamics, so that the closed-loop system converges to the economically optimal operating point without relying on disturbance forecasts. Based on this formulation, we develop a Data-Enabled Policy Optimization (DeePO)-based online learning controller and incorporate Adaptive Moment Estimation (ADAM) to improve closed-loop performance. We further establish convergence and performance guarantees for the resulting closed-loop system. Simulations on an industrial-park DHS in Northern China show that the proposed method achieves stable near-optimal operation and strong empirical robustness to both static and time-varying model mismatch under practical disturbance conditions.
\end{abstract}

\begin{highlights}
\item Optimality conditions are embedded into augmented DHS dynamics.
\item A data-driven controller is developed for real-time DHS operation.
\item Adaptive online updates improve learning and closed-loop performance.
\item The framework guarantees controller optimality and closed-loop convergence.
\item Validation on an industrial-park DHS shows stable near-optimal operation.
\end{highlights}

\begin{keywords}
District heating systems \sep Economic dispatch \sep Temperature regulation \sep Online data-driven control\sep Performance guarantees
\end{keywords}
\maketitle

\section{Introduction}

Heating systems account for a significant share of global energy consumption and greenhouse gas emissions. Improving the operational efficiency and flexibility of district heating systems (DHSs) is therefore important for low-carbon energy transitions. Widely deployed in China, Russia, and Europe, DHSs distribute thermal energy through large-scale pipeline networks. Their increasing integration with renewable and waste-heat sources reduces reliance on fossil fuels, but also increases operational uncertainty and coordination complexity~\cite{liu2024diversifying}.

In practical DHS operation, coordinating economic dispatch and temperature regulation remains challenging under demand uncertainty and model mismatch. Existing forecast-based and model-based control strategies can perform well when disturbance predictions and nominal models are accurate, but their performance may degrade under uncertain heat demand and changing operating conditions. This motivates a closer examination of temperature regulation for DHSs under uncertainty.

Although temperature regulation has been widely studied in buildings~\cite{cholewa2022easy}, DHSs differ substantially in system structure, control objectives, and operating conditions. Compared with building heating control, DHS regulation requires coordinated heat generation, transport, and allocation over large-scale networks, where control performance depends on network interconnections and thermal transport dynamics~\cite{machado2022decentralized,ahmed2023control}. This motivates the development of control frameworks for DHSs that can maintain economically efficient and thermally stable operation under uncertainty.

\subsection{Forecast-reliant and model-based DHS operation}
In practice, DHS operation often combines forecast-based \textit{setpoint scheduling} with real-time \textit{setpoint tracking}\cite{machado2022decentralized,ahmed2023control}. However, forecast errors and model mismatch can degrade both thermal and economic performance under uncertain heat demand~\cite{frison2024evaluating,jansen2024effect}. Model predictive control (MPC) partially alleviates this issue through receding-horizon re-optimization~\cite{jansen2024mixed}. Related DHS operation frameworks have been developed using improved load forecasting methods~\cite{wei2022data,wei2024investigation} together with control-oriented or reduced-order thermal models~\cite{saloux2021modelbased,la2023optimal,zhang2023new} to improve thermal and economic performance. However, their practical performance remains sensitive to forecast quality and model accuracy, especially when heat demand varies and model mismatch is present.

\subsection{Data-driven enhancements within MPC-based DHS operation}
Building on these MPC frameworks, recent studies have further introduced data-driven components to mitigate the impact of model mismatch in DHS operation. However, rather than removing the reliance on model knowledge, these methods are still largely developed within MPC frameworks. For example, recent works have explored data-driven optimization frameworks within MPC~\cite{wei2022data} and advanced control schemes combining MPC with learning-based approaches~\cite{mugnini2022advanced}. Nevertheless, these approaches remain prediction-reliant and model-reliant: their performance still depends on the quality of the predictive model and its associated information, while formal closed-loop guarantees under prediction errors and model mismatch remain limited. These limitations motivate a shift from prediction-reliant MPC frameworks towards online data-driven control frameworks that can support economically consistent and reliable DHS operation under uncertainty while providing stronger closed-loop guarantees. Feedback-based optimization controllers~\cite{hauswirth2024optimization} provide an alternative, but their fast-optimization assumption may be restrictive for slow thermal dynamics. Our previous linear quadratic regulator (LQR) framework~\cite{yi2025optimal} guarantees convergence to the economically optimal operating point under unknown deterministic disturbances, but does not account for stochastic demand variations or renewable-heat fluctuations encountered in practical operation.

\subsection{Online data-driven control opportunities for real-time economic dispatch and temperature regulation in DHSs}
Data-driven LQR provides a promising framework for online DHS control when accurate system models are difficult to maintain and strong closed-loop guarantees are desired. Existing approaches broadly fall into two categories: indirect methods, which first identify system dynamics and then solve a certainty-equivalent (CE) control problem, and direct methods, which learn the control policy directly from data. For large-scale DHSs under uncertainty and model mismatch, direct methods are particularly attractive because they reduce reliance on repeated model identification and can adapt more naturally to closed-loop operating data.

Representative direct data-driven LQR methods include gradient-based policy updates~\cite{tu2019sample,mohammadi2021convergence} and Data-Enabled Policy Optimization (DeePO)~\cite{zhao2025data}. Gradient-based methods often require relatively long data trajectories to estimate gradients or value functions accurately, which can limit online sample efficiency. By contrast, DeePO constructs a covariance-based surrogate objective directly from closed-loop data, making it well suited to online adaptation from a single operating trajectory. This is particularly relevant for DHS operation, where repeated restarts, multiple independent experiments, or extensive exploration are typically impractical.

Despite these advantages, existing DeePO studies have mainly been validated on low-dimensional benchmark systems, and its online implementation can still be sensitive to stochastic disturbances, time-correlated data, and surrogate bias under model mismatch. To improve performance in this setting, we incorporate Adaptive Moment Estimation (ADAM)~\cite{kingma2014adam} into the DeePO update. Although ADAM is widely used to accelerate gradient-based learning~\cite{cui2025fully,kim2025neural}, its convergence properties for LQR policy learning do not follow directly from existing general-purpose results, because the cost depends implicitly on Lyapunov equations. This motivates an ADAM-enhanced DeePO framework for optimal real-time economic dispatch and temperature regulation in DHSs under demand uncertainty and model mismatch.

\subsection{Contributions and paper organization}
The main contributions of this work are threefold:
\begin{itemize}
    \item \textbf{Economically consistent online control formulation for DHS operation.}
We develop a data-driven online control framework for real-time optimal economic dispatch and temperature regulation in DHSs. Steady-state economic optimality conditions are embedded into the DHS temperature dynamics, yielding an augmented regulation problem whose closed-loop equilibrium coincides with the economically optimal operating point, without relying on disturbance forecasts for control design or real-time operation.

    \item \textbf{Online data-driven control with convergence and stability guarantees.}
    Based on Data-Enabled Policy Optimization (DeePO), we develop an online data-driven controller for DHSs under stochastic disturbances and model mismatch, and establish convergence to an optimal control policy, together with closed-loop stability guarantees.
    
    \item \textbf{ADAM-enhanced online policy learning with improved closed-loop performance.}
    We incorporate Adaptive Moment Estimation (ADAM) into the DeePO update to improve closed-loop performance, and establish convergence guarantees for the resulting ADAM-enhanced scheme in large-scale DHS control.
\end{itemize}

The remainder of the paper is organized as follows. Section~2 presents the DHS model. Sections~3 and~4 present a control-oriented formulation for DHS operation and the corresponding online data-driven controller design. Section~5 presents simulation results on an industrial-scale DHS with model mismatch and time-varying parameter perturbations. Section~6 concludes the paper.

\section{District heating system (DHS) model}
\subsection{Temperature dynamics of DHSs}
District heating system (DHS) temperature dynamics can be modeled at different fidelity levels. High-fidelity partial differential equation (PDE) models capture spatiotemporal heat transport in detail, but are often impractical for real-time optimization and learning-based control because of their high dimensionality and nonlinear structure~\cite{simonsson2021towards}. Reduced-order models that emphasize aggregate heat transport and node-level mixing are therefore more suitable for scalable network-level control~\cite{machado2022decentralized,la2023optimal,qin2024frequency,ahmed2023control}. Representative DHS pipeline models are summarized in Table~\ref{tab:comparison}. We therefore adopt a control-oriented DHS temperature model based on energy conservation and mass-flow mixing, which yields a low-dimensional state-space representation while preserving the network topology and interconnection structure.

\begin{table}[htbp]
\centering
\caption{Comparison of representative DHS pipeline models. Here, $T_k^i$ denotes the outlet temperature of the $i^{\text{th}}$ pipe segment at time step $k$, $\boldsymbol{T}^{\mathrm{in}}$ and $\boldsymbol{T}$ denote the inlet and outlet temperature vectors, respectively, $\boldsymbol{v}$ is the diagonal flow-velocity matrix, $T_a$ is the ambient temperature, and $\varrho$, $c_p$, $\boldsymbol{\lambda}_a$, $A_s$, and $\boldsymbol{V}$ denote the fluid density, specific heat capacity, heat-transfer coefficient, pipe cross-sectional area, and pipe volume, respectively. $f_k^i$ is the heat loss in the $i^{\text{th}}$ pipe segment at time step $k$.}
\begin{tabular}{p{2cm}|p{5cm}}
\toprule
\textbf{Reference} & \textbf{Model} \\ 
\midrule
\cite{simonsson2021towards}
& $\dfrac{T_{k+1}^{i}-T_k^i}{\Delta t}
=-v\dfrac{T_k^i-T_{k}^{i-1}}{\Delta x}
-\dfrac{f_k^i}{\varrho c_p A_s}$ \\
\midrule
\cite{ahmed2023control}
& $\boldsymbol{V}\dot{\boldsymbol{T}}
=\boldsymbol{v}(\boldsymbol{T}^{\mathrm{in}}-\boldsymbol{T})
+\boldsymbol{\lambda}_a(T_a-\boldsymbol{T})$ \\
\midrule
\cite{machado2022decentralized,la2023optimal,qin2024frequency}
& $\boldsymbol{V}\dot{\boldsymbol{T}}
=\boldsymbol{v}(\boldsymbol{T}^{\mathrm{in}}-\boldsymbol{T})$ \\
\bottomrule
\end{tabular}
\label{tab:comparison}
\end{table}

Let $\mathcal{H}^E$ denote the set of edges (heat exchangers and pipelines) and $\mathcal{H}^N$ the set of nodes (storage tanks).  
The edge dynamics in~(\ref{eq:2a}) capture the dependence of outlet temperatures on inlet--outlet temperature differences,  
while the node dynamics in~(\ref{eq:2b}) describe the balance between variations in stored thermal energy and the net inflow from adjacent edges, where $T_j^E$ is the outlet temperature of edge $j$ and $T_k^N$ is the temperature of node $k$:
\begin{subequations}\label{temdy}
\begin{align}
\varrho C_p V_j^E \dot{T}_j^E =& \varrho C_p q_j^E (T_k^N - T_j^E) + h_j^G - h_j^L, \notag\\
& j \in \mathcal{H}^E, k \in \mathcal{H}^N, \label{eq:2a} \\
\varrho C_p V_k^N \dot{T}_k^N =& \sum_{j \in \mathcal{T}_k} \varrho C_p q_j^E (T_j^E - T_k^N), 
k \in \mathcal{H}^N. \label{eq:2b}
\end{align}
\end{subequations}
$h_j^G$ and $h_j^L$ denote the heat-source and heat-load powers at edge $j$, respectively. $\mathcal{T}_k$ is the set of edges whose outlets are node $k$.
The parameters $V_j^E$ and $V_k^N$ represent edge and node volumes, $q_j^E$ is the mass-flow rate, and $C_p$ and $\varrho$ are the specific heat capacity and density of water. 

The edge set is partitioned into the heat producer set $\boldsymbol{\mathcal G}$, heat load set $\boldsymbol{\mathcal L}$, and pipeline set $\boldsymbol{\mathcal P}$. For an edge $j\in\boldsymbol{\mathcal G}$, we have $h_j^L = 0$. For an edge $j\in\boldsymbol{\mathcal L}$, we have $h_j^G = 0$ and $h_j^L$ is prescribed. For an edge $j\in\boldsymbol{\mathcal P}$, we have $h_j^G = 0$ and $h_j^L = 0$.

Equations~(\ref{eq:2a})–(\ref{eq:2b}) can be compactly expressed in matrix form as \footnote{$\boldsymbol{A_h}$ satisfies $\boldsymbol{A_h \mathbf{1} = 0}$, $\boldsymbol{\mathbf{1}^\top A_h = 0}$, and $\boldsymbol{A_h + A_h^\top \succeq 0}$ with a simple zero eigenvalue, hence it can be regarded as a Kirchhoff matrix of the heating network. Therefore, the null space of $\boldsymbol{A_h}$ and  $\boldsymbol{A_1}$ is $z\boldsymbol{1}$, where $z \in \mathbb{R}$.  }.
\begin{subequations}
\begin{align}
&\boldsymbol{V}
\begin{bmatrix}
\dot{\boldsymbol{T}}^G \\
\dot{\boldsymbol{T}}^{L}\\
\dot{\boldsymbol{T}}^N
\end{bmatrix}
= -\boldsymbol{A}_h 
\begin{bmatrix}
\boldsymbol{T}^G \\
\boldsymbol{T}^{L}\\
\boldsymbol{T}^N
\end{bmatrix}
+ 
\begin{bmatrix}
\boldsymbol{h^G} \\
-\boldsymbol{h}^{L} \\
\boldsymbol{0}
\end{bmatrix}\label{eq:3},\\
&\boldsymbol{V}
\begin{bmatrix}
\dot{\boldsymbol{T}}^E \\
\dot{\boldsymbol{T}}^N
\end{bmatrix}
= -\boldsymbol{A}_h 
\begin{bmatrix}
\boldsymbol{T}^E \\
\boldsymbol{T}^N
\end{bmatrix}
+ 
\begin{bmatrix}
\begin{bmatrix}
\boldsymbol{h^G} \\
-\boldsymbol{h}^{L}
\end{bmatrix}\\
\boldsymbol{0}
\end{bmatrix}
\label{eq:e3},\\
&\boldsymbol{\dot{T}}=\boldsymbol{-A_1T+B_1h^G-B_2 h^L},\label{origindy}
\end{align}
\end{subequations}
where $\boldsymbol{h^G}$ and $\boldsymbol{h^L}$ denote the production and load vectors, respectively, scaled by $\frac{1}{\varrho C_p}$.

$\small{\boldsymbol{A_1 = V^{-1}A_h},
\boldsymbol{B_1 = V^{-1}\begin{bmatrix} I \\ 0 \\ 0 \end{bmatrix}}, 
\boldsymbol{B_2 = V^{-1}\begin{bmatrix} 0 \\ I \\ 0 \end{bmatrix}}}$. $\boldsymbol{A_h}$ is defined as $\boldsymbol{A_h} =
\begin{bmatrix}
\operatorname{diag}(\boldsymbol{q^E}) & -\operatorname{diag}(\boldsymbol{q^E}) \boldsymbol{B_{sh}} \\
\boldsymbol{-B_{th}} \operatorname{diag}(\boldsymbol{q^E}) & \operatorname{diag}(\boldsymbol{B_{th} q^E})
\end{bmatrix}$.
$\boldsymbol{A_h}$ is a constant matrix for a given mass flow vector $q^E$, where  
$\boldsymbol{B_{th}} = \tfrac{1}{2}(|\boldsymbol{B_h}| + \boldsymbol{B_h})$, and 
$\boldsymbol{B_{sh}} = \tfrac{1}{2}(|\boldsymbol{B_h}| - \boldsymbol{B_h})$, 
where $\boldsymbol{B_h}$ is the incidence matrix of the DHS, and $|\boldsymbol{B_h}|$ denotes the elementwise absolute value of $\boldsymbol{B_h}$~\cite{qin2024frequency}.

\subsection{Steady-state optimization of DHSs}
Given the steady-state heat demand $\boldsymbol{\bar h}^L$, the steady-state DHS operation is characterized by the following two optimization problems. Problem \textbf{E1} determines the economically optimal heat generation, while problem \textbf{E2} selects the corresponding temperature profile by minimizing temperature deviation over the equilibrium set.
\begin{subequations}\label{opt1}
\begin{align}
\textbf{E1:} &\min_{\boldsymbol{h^G}\in\mathbb{R}^{|\mathcal G|}, \quad \boldsymbol{T}\in\mathbb{R}^{n_T}} \frac{1}{2} \boldsymbol{h^G}^{\top} \boldsymbol{F^G} \boldsymbol{h^G},\label{5a}\\
    &\text{s.t.} \boldsymbol{A_1 T} = \boldsymbol{B_1}\boldsymbol{h^G-B_2 \bar{h}^L},\label{5b}
\end{align}
\end{subequations}
where $\boldsymbol{F^G}=\mathrm{diag}(F_i^G)\succ\boldsymbol{0}$ collects the producer cost coefficients for $i\in\boldsymbol{\mathcal G}$. \textbf{E1} admits a unique optimizer $\boldsymbol{h}^{G\star}$. The associated equilibrium temperature set is
$\boldsymbol{\tilde T}
=
\boldsymbol{A}_1^\dagger(\boldsymbol{B}_1\boldsymbol{h}^{G\star}-\boldsymbol{B}_2\boldsymbol{\bar h}^L)
+z\boldsymbol{1}$, $z\in\mathbb R$,
where $\boldsymbol{A}_1^\dagger$ denotes the Moore--Penrose pseudoinverse. Problem \textbf{E2} then determines $\boldsymbol{T}^\star$ by minimizing the temperature-deviation cost over this set.
\begin{subequations}\label{economic2}
\begin{align}
\textbf{E2:} &\min_{z \in \mathbb{R}, \boldsymbol{T}\in\mathbb{R}^{n_T}} \frac{1}{2} \boldsymbol{T}^\top \boldsymbol{F^D} \boldsymbol{T},\label{4a}\\
&s.t. \boldsymbol{T} = \boldsymbol{A_1^\dagger} (\boldsymbol{B_1} \boldsymbol{{h^G}^\star} - \boldsymbol{B_2\bar{h}^L}) + z\boldsymbol{1},\label{4b}
\end{align}
\end{subequations}
where $\boldsymbol{F^D}=diag (F^D_{i})\succ\boldsymbol{0}$, and $F^D_{i}$ represents the temperature deviation cost coefficient at node $i$.

The following result characterizes the steady-state optimality conditions that link the DHS equilibrium to problems \textbf{E1} and \textbf{E2}.

\begin{theorem}\label{optimalitycondition}
If the DHS (\ref{origindy}) achieves an equilibrium at $\boldsymbol{{T^\star}}$ and $\boldsymbol{h^G}^\star$, and satisfies $\boldsymbol{F^M {h^G}^\star} = \boldsymbol{0}$ and $\boldsymbol{1^\top F^D {T^\star}} = 0$, where $\boldsymbol{F^M}\in\mathbb{R}^{(|\mathcal G|-1)\times |\mathcal G|}$ is defined by
\[
\boldsymbol{F^M}=
\begin{bmatrix}
F_1^G & -F_2^G & 0 & \cdots & 0\\
0 & F_2^G & -F_3^G & \cdots & 0\\
\vdots & \vdots & \ddots & \ddots & \vdots\\
0 & 0 & \cdots & F_{|\mathcal G|-1}^G & -F_{|\mathcal G|}^G
\end{bmatrix}.
\]
Then it uniquely solves the optimization problems \textbf{E1} and \textbf{E2}.
\end{theorem}
The proof is deferred to Appendix~\ref{app:optimalitycondition}.

\section{Control-oriented formulation for DHS operation}

In this paper, $\|\cdot\|$ denotes the Euclidean norm, $\|\cdot\|_2$ the induced $L_2$ norm, $\|\cdot\|_F$ the Frobenius norm, $\|\cdot\|_\star$ the nuclear norm, and $\|\cdot\|_{\mathcal H_2}$ the $\mathcal H_2$ norm. The notation $\langle A,B\rangle:=\mathrm{Tr}(A^\top B)$ denotes the Frobenius inner product.

\subsection{Discrete-time temperature dynamics}
To facilitate controller design and simulation, we discretize the continuous-time dynamics in~\eqref{origindy}. For simplicity, we use Euler discretization with sampling interval $\tau$. The controller developed in the remainder of the paper is not restricted to this choice and can be applied with any discretization method that yields discrete-time dynamics with a constant sampling interval. The resulting discrete-time DHS model is
\begin{equation}\label{disheat}
\begin{aligned}
\boldsymbol{T}_{k+1}
=&\,
(\boldsymbol{I}-\tau\boldsymbol{A}_1)\boldsymbol{T}_k
+\tau\boldsymbol{B}_1\boldsymbol{h}^G_k
-\tau\boldsymbol{B}_2(\boldsymbol{h}^{L,\mathrm{con}}+\boldsymbol{h}^{L,\mathrm{sto}}_k)\\
=&\,
\boldsymbol{A}_T\boldsymbol{T}_k
+\boldsymbol{B}_T\boldsymbol{h}^G_k
+\boldsymbol{B}_T^L\boldsymbol{h}^{L,\mathrm{con}}
+\boldsymbol{B}_T^L\boldsymbol{h}^{L,\mathrm{sto}}_k,
\end{aligned}
\end{equation}
where the heat load $\boldsymbol{h}^L$ is decomposed into a slow-varying component $\boldsymbol{h}^{L,\mathrm{con}}$ and a fast-varying stochastic component $\boldsymbol{h}^{L,\mathrm{sto}}_k$.

\subsection{Output and error definition}
We aim to design a temperature regulator that ensures convergence to the optimal equilibrium point $(\boldsymbol{T^\star}, \boldsymbol{h^{G\star}})$ defined by the solutions of \textbf{E1} and \textbf{E2}. To achieve this, we introduce an error signal whose convergence to zero guarantees satisfaction of the corresponding optimality conditions. Specifically, the error definition is constructed directly from the \textbf{E1–E2} optimality conditions in \textbf{Theorem~\ref{optimalitycondition}}, with $\boldsymbol{e_k}$ having the same dimension as $\boldsymbol{h^G_k}$:
\begin{equation}\label{ll13}
\boldsymbol{e_k} = \begin{bmatrix} \boldsymbol{0} \\ \boldsymbol{1^\top F^D} \end{bmatrix}\boldsymbol{T_k}+\begin{bmatrix} \boldsymbol{F^M} \\ \boldsymbol{0} \end{bmatrix} \boldsymbol{h^G_k}=\boldsymbol{C_TT_k+D_Th^G_k}.
\end{equation}

\subsection{Augmented dynamics}
To achieve both temperature regulation and economically consistent steady-state operation, we design the controller so that the optimality error satisfies $\mathbb{E}[\boldsymbol{e}_k]\to\boldsymbol{0}$ and the heat-generation increment satisfies $\mathbb{E}[\boldsymbol{h}_k^G-\boldsymbol{h}_{k-1}^G]\to\boldsymbol{0}$ at equilibrium in expectation. This motivates the augmented state:
$\boldsymbol{x_{k+1}}=\begin{bmatrix}
    \boldsymbol{T_{k+1}-T_k}\\\boldsymbol{e_k}
\end{bmatrix}$, the resulting augmented dynamics are described as follows,
\begin{subequations}\label{augmentd}
\begin{align}
\boldsymbol{x_{k+1}}=&\begin{bmatrix}
    \boldsymbol{A_T}&\boldsymbol{0}\\
    \boldsymbol{C_T}&\boldsymbol{I}
\end{bmatrix}\boldsymbol{x_k}+\begin{bmatrix}
    \boldsymbol{B_T}\\\boldsymbol{D_T}
\end{bmatrix}\boldsymbol{u_k}+\begin{bmatrix}
    \boldsymbol{B_T^L}\\\boldsymbol{0}
\end{bmatrix}\boldsymbol{w_k}\label{9a}\\
=&\boldsymbol{A x_k+B  u_k+B_w w_k}\notag\\
\boldsymbol{u_k}=&\boldsymbol{h^G_k-h^G_{k-1}},\boldsymbol{w_k}=\boldsymbol{h}^{L,\mathrm{sto}}_k-\boldsymbol{h}^{L,\mathrm{sto}}_{k-1},\\
\boldsymbol{e_k}=&\begin{bmatrix}
    \boldsymbol{C_T}&\boldsymbol{I}
\end{bmatrix}\boldsymbol{x_k+D_Tu_k}\\\notag
=&\boldsymbol{Cx_k+D u_k}.
\end{align}
\end{subequations} 

\begin{assumption}\label{ass:aug_cont} The pair $(A,B)$ in \eqref{augmentd} is controllable. \end{assumption}

\begin{assumption}\textbf{(Bounded disturbances with finite covariance)}
\label{ass:noise_bounded}
The disturbance sequence $\{\boldsymbol{w}_k\}$ is zero-mean and uniformly bounded.
Specifically, there exists a constant $\delta>0$ such that
$\|\boldsymbol{w}_k\|\le\delta$ for all $k$.
Moreover, the disturbance has a finite, time-invariant covariance matrix, i.e.,
$\mathbb{E}[\boldsymbol{w}_k]=\boldsymbol{0},
\mathbb{E}[\boldsymbol{w}_k\boldsymbol{w}_k^\top]
=\boldsymbol{U}^{\mathrm{dis}}\succ\boldsymbol{0},
$
where $\boldsymbol{U}^{\mathrm{dis}}$ is finite.
\end{assumption}

\begin{definition}
    The augmented system (\ref{augmentd}) has an input dimension of $m=n_{h^G_k}$ and a state dimension of $n=n_{T_k}+n_{h^G_k}$.
\end{definition}

\begin{proposition}\textbf{(Convergence to the optimal operating point under a stabilizing feedback law)}
\label{prop:mean_secondmoment}
Consider the original DHS (\ref{origindy}) and its augmented representation in~(\ref{augmentd}) under the state-feedback law
$\boldsymbol{u}_k = \boldsymbol{K}\boldsymbol{x}_k$,
where $\boldsymbol{K}$ is \emph{stabilizing}.
If
$\mathbb{E}[\boldsymbol{x}_k]\to\boldsymbol{0}$,
$\mathbb{E}[\boldsymbol{u}_k]\to\boldsymbol{0}$, and
$\mathbb{E}[\boldsymbol{e}_k]\to\boldsymbol{0}$,
then the original DHS converges in expectation to its optimal equilibrium
$\bigl(\boldsymbol{T^\star}, \boldsymbol{h^{G\star}}\bigr)$ defined in~(\ref{opt1},\ref{economic2}).
Moreover, under \textbf{Assumption~\ref{ass:noise_bounded}}, the closed-loop state is mean-square bounded, i.e.,
$\sup_k \mathbb{E}\|\boldsymbol{x}_k\|^2<\infty$; equivalently,
$\mathbb{E}[\boldsymbol{x}_k\boldsymbol{x}_k^\top]$ converges to a unique stationary covariance that scales with the disturbance covariance level.
\end{proposition}

\begin{pf}
From $\mathbb{E}[\boldsymbol{u}_k]=\mathbb{E}[\boldsymbol{h}^G_k-\boldsymbol{h}^G_{k-1}]=\boldsymbol{0}$ it follows that
$\mathbb{E}[\boldsymbol{h}^G_k]$ is constant for $k$ sufficiently large.
Together with $\mathbb{E}[\boldsymbol{x}_k]\to\boldsymbol{0}$ and $\mathbb{E}[\boldsymbol{e}_k]\to\boldsymbol{0}$, the expected steady-state of the original DHS satisfies the KKT-based optimality conditions characterized in \textbf{Theorem~\ref{optimalitycondition}}; hence the original DHS converges in expectation to
$\bigl(\boldsymbol{T^\star}, \boldsymbol{h^{G\star}}\bigr)$. Since $\boldsymbol{K}$ is stabilizing, the closed-loop matrix satisfies $\rho(\boldsymbol{A}+\boldsymbol{B}\boldsymbol{K})<1$,
where $\rho(\cdot)$ denotes the spectral radius.
With zero-mean disturbances of finite covariance, the standard discrete-time Lyapunov argument implies mean-square boundedness and convergence of $\mathbb{E}[\boldsymbol{x}_k\boldsymbol{x}_k^\top]$ to the stationary covariance.
\end{pf}

\subsection{Model-based LQR benchmark}
\subsubsection{Stationary distribution of $\boldsymbol{x_k}$}
To support online learning, we inject a bounded exploration signal with magnitude $\sigma>0$. Let $\{\boldsymbol{w}_{s,k}\}$ be a zero-mean bounded sequence with $\mathbb{E}[\boldsymbol{w}_{s,k}] = \boldsymbol{0}$ and $\mathbb{E}[\boldsymbol{w}_{s,k}\boldsymbol{w}_{s,k}^\top] = \boldsymbol{I}_k$. The control input is $\boldsymbol{u}_k = \boldsymbol{K}\boldsymbol{x}_k + \sigma \boldsymbol{w}_{s,k}$, yielding the closed-loop dynamics $\boldsymbol{x}_{k+1} = (\boldsymbol{A}+\boldsymbol{B}\boldsymbol{K})\boldsymbol{x}_k+\boldsymbol{\epsilon}_k$, where $\boldsymbol{\epsilon}_k:=\boldsymbol{B_w w}_k+\sigma \boldsymbol{B w}_{s,k}$.

Under Assumption~\ref{ass:noise_bounded}, $\{\boldsymbol{w}_k\}$ and $\{\boldsymbol{w}_{s,k}\}$ are zero-mean and independent, so $\boldsymbol{\epsilon}_k$ is zero-mean with covariance $\mathbb{E}[\boldsymbol{\epsilon}_k \boldsymbol{\epsilon}_k^\top]$ $=\boldsymbol{B_w U^{dis} B_w^\top}+\sigma^2 \boldsymbol{B B^\top}=: \boldsymbol{U_\epsilon}\succ \boldsymbol{0}$. Since both sequences are bounded, $\{\boldsymbol{\epsilon}_k\}$ is also bounded.

For $\rho(\boldsymbol{A}+\boldsymbol{B}\boldsymbol{K})<1$, the closed-loop state process admits a stationary covariance $\boldsymbol{U}_K$, which is the unique positive semidefinite solution to
\begin{subequations}
\begin{align}
\boldsymbol{U_K}
&=
\boldsymbol{U_\epsilon}
+
(\boldsymbol{A}+\boldsymbol{B}\boldsymbol{K})\,\boldsymbol{U_K}\,(\boldsymbol{A}+\boldsymbol{B}\boldsymbol{K})^\top, 
\label{eq:UK_Lyap}
\\
\boldsymbol{U_K}
&=
\sum_{i=0}^\infty 
(\boldsymbol{A}+\boldsymbol{B}\boldsymbol{K})^i
\,\boldsymbol{U_\epsilon}\,
[(\boldsymbol{A}+\boldsymbol{B}\boldsymbol{K})^\top]^i.
\label{eq:UK_series}
\end{align}
\end{subequations}
It is a standard result for stable discrete-time stochastic linear systems~\cite{zhou1996robust}.
\subsubsection{Cost function}
To quantify long-run regulation performance, we define the performance output as
$\boldsymbol{z}_k=
\begin{bmatrix}
\boldsymbol{Q}^{1/2}\boldsymbol{x}_k \\
\boldsymbol{R}^{1/2}\boldsymbol{u}_k
\end{bmatrix}$,
where $\boldsymbol{Q}\succ\boldsymbol{0}$ and $\boldsymbol{R}\succ\boldsymbol{0}$ weight state deviations and control effort, respectively. We consider the feedback law $\boldsymbol{u}_k=\boldsymbol{K}\boldsymbol{x}_k$.

Under Assumption~\ref{ass:noise_bounded} and a stabilizing gain $\boldsymbol{K}$ satisfying $\rho(\boldsymbol{A}+\boldsymbol{B}\boldsymbol{K})<1$, the closed-loop system admits a unique stationary covariance. The corresponding $\mathcal{H}_2$ cost of the closed-loop map $\mathcal{T}(\boldsymbol{K}):\boldsymbol{B}_w\boldsymbol{w}\mapsto\boldsymbol{z}$ is
\begin{subequations}\label{eqcost}
\begin{align}
C(\boldsymbol{K})
:=\|\mathcal{T}(\boldsymbol{K})\|_{\mathcal H_2}^2
&= \mathrm{Tr}\!\bigl((\boldsymbol{Q}+\boldsymbol{K}^\top\boldsymbol{R}\boldsymbol{K})\,\boldsymbol{U}_K\bigr),
\label{eqcosta}\\
&= \mathrm{Tr}\!\bigl(\boldsymbol{P}_K\,\boldsymbol{U}_\epsilon\bigr),
\label{eqcostb}
\end{align}
\end{subequations}
where $\boldsymbol{P}_K\succ\boldsymbol{0}$ is the unique solution to
\begin{equation}\label{lyapunovdefine}
\boldsymbol{P}_K
=
\boldsymbol{Q}
+
\boldsymbol{K}^\top\boldsymbol{R}\boldsymbol{K}
+
(\boldsymbol{A}+\boldsymbol{B}\boldsymbol{K})^\top
\boldsymbol{P}_K
(\boldsymbol{A}+\boldsymbol{B}\boldsymbol{K}).
\end{equation}
Equivalently, $C(\boldsymbol{K})$ is the infinite-horizon long-run average quadratic cost
$C(\boldsymbol{K})=
\lim_{N\to\infty}\frac{1}{N}\sum_{k=0}^{N-1}\mathbb{E}[\boldsymbol{z}_k^\top\boldsymbol{z}_k]$.

\subsubsection{Model-based LQR}
When the system matrices $\boldsymbol{(A,B)}$ are known and the pair is stabilizable, the infinite-horizon discrete-time LQR problem admits the optimal state-feedback solution
\begin{equation}
\boldsymbol{K^\star}
=
-(\boldsymbol{R}+\boldsymbol{B}^\top \boldsymbol{P^\star}\boldsymbol{B})^{-1}
\boldsymbol{B}^\top \boldsymbol{P^\star}\boldsymbol{A},
\end{equation}
where $\boldsymbol{P^\star}\succ\boldsymbol{0}$ is the stabilizing solution to the discrete-time algebraic Riccati equation. This standard model-based solution serves as a benchmark for the subsequent data-driven controller design~\cite{anderson2007optimal}.

\subsubsection{Indirect certainty-equivalence (CE) LQR}
As a benchmark data-driven approach, CE first identifies a model from data and then solves the corresponding LQR problem as if the identified dynamics were exact. Consider a trajectory of length $t$ generated under a stabilizing policy $\boldsymbol{K}_0$, and define
\begin{equation}\label{collect}
\begin{aligned}
\boldsymbol{X}_0 &:= [\boldsymbol{x}_0 \quad \boldsymbol{x}_1 \quad \cdots \quad \boldsymbol{x}_{t-1}] \in \mathbb{R}^{n \times t}, \\
\boldsymbol{U}_0 &:= [\boldsymbol{u}_0 \quad \boldsymbol{u}_1 \quad \cdots \quad \boldsymbol{u}_{t-1}] \in \mathbb{R}^{m \times t}, \\
\boldsymbol{X}_1 &:= [\boldsymbol{x}_1 \quad \boldsymbol{x}_2 \quad \cdots \quad \boldsymbol{x}_t] \in \mathbb{R}^{n \times t},\\
\boldsymbol{W}_0 &:= [\boldsymbol{w}_0 \quad \boldsymbol{w}_1 \quad \cdots \quad \boldsymbol{w}_{t-1}] \in \mathbb{R}^{n \times t}.
\end{aligned}
\end{equation}
Let $\boldsymbol{D}_0:=\begin{bmatrix}\boldsymbol{U}_0\\ \boldsymbol{X}_0\end{bmatrix}= [\boldsymbol{d}_0 \quad \boldsymbol{d}_1 \quad \cdots \quad \boldsymbol{d}_{t-1}]$ and $\boldsymbol{\Phi}:=\frac{1}{t}\boldsymbol{D}_0\boldsymbol{D}_0^\top$, so that $\boldsymbol{X}_1 = \boldsymbol{A X}_0 + \boldsymbol{B U}_0 + \boldsymbol{W}_0$.

\begin{assumption}[Persistent excitation]
\label{ass:PEo}
The input sequence $\{\boldsymbol{u}_k\}$ is persistently exciting of order $n+m$, so that the data matrix $\boldsymbol{D}_0$ has full row rank.
\end{assumption}

Under Assumption~\ref{ass:PEo}, the least-squares estimate is
\begin{equation}\label{lseor}
[\boldsymbol{\hat{B}},\,\boldsymbol{\hat{A}}]
=
\boldsymbol{\bar{X}}_1\boldsymbol{\Phi}^{-1},
\qquad
\boldsymbol{\bar{X}}_1:=\frac{1}{t}\boldsymbol{X}_1\boldsymbol{D}_0^\top.
\end{equation}
The aggregated-disturbance covariance is estimated by
\begin{equation}\label{matrixestimate}
\boldsymbol{\hat{\varepsilon}}_0
=
\boldsymbol{X}_1-(\boldsymbol{\hat{A}}+\boldsymbol{\hat{B}}\boldsymbol{K}_0)\boldsymbol{X}_0,
\qquad
\boldsymbol{\hat{U}}_\epsilon
=
\frac{1}{t}\boldsymbol{\hat{\varepsilon}}_0\boldsymbol{\hat{\varepsilon}}_0^\top.
\end{equation}
Fixing $\boldsymbol{\hat{U}}_\epsilon$ from the initial batch, the CE LQR problem is
\begin{subequations}\label{eq:ceop_u}
\begin{align}
\min_{\boldsymbol{K},\boldsymbol{\hat{U}}_K}\quad
&C_{CE}(\boldsymbol{K})
=
\mathrm{Tr}\!\left((\boldsymbol{Q}+\boldsymbol{K}^\top\boldsymbol{R}\boldsymbol{K})\boldsymbol{\hat{U}}_K\right),\label{eq:ceop_ua}\\
\text{s.t.}\quad
&\boldsymbol{\hat{U}}_K
=
\boldsymbol{\hat{U}}_\epsilon
+
(\boldsymbol{\hat{A}}+\boldsymbol{\hat{B}}\boldsymbol{K})
\boldsymbol{\hat{U}}_K
(\boldsymbol{\hat{A}}+\boldsymbol{\hat{B}}\boldsymbol{K})^\top.
\label{eq:ceop_ub}
\end{align}
\end{subequations}

\begin{remark}[Role of estimating $\boldsymbol{U}_\epsilon$]
Although the optimal controller is theoretically independent of $\boldsymbol{U}_\epsilon$, estimating $\boldsymbol{U}_\epsilon$ improves the physical relevance of covariance-weighted learning in DHS applications, where disturbances may be anisotropic and state-correlated.
\end{remark}

\section{Data-enabled policy optimization (DeePO)}

\subsection{Covariance parameterization}

\begin{assumption}[Uniform boundedness of signals]
\label{ass:bounded_signals}
There exist positive constants $\bar x>0$, $\bar u>0$, and $\bar d>0$ such that $\|\boldsymbol{x}_k\|\le \bar x$, $\|\boldsymbol{u}_k\|\le \bar u$, and $\|\boldsymbol{d}_k\|\le \bar d$ for all $k\in\mathbb{N}$.
\end{assumption}

To obtain a data-enabled parameterization of the LQR problem, define
$\boldsymbol{\bar{U}}_0 := \frac{1}{t}\boldsymbol{U}_0 \boldsymbol{D}_0^\top$,
$\boldsymbol{\bar{X}}_0 := \frac{1}{t}\boldsymbol{X}_0 \boldsymbol{D}_0^\top$,
$\boldsymbol{\bar{X}}_1 := \frac{1}{t}\boldsymbol{X}_1 \boldsymbol{D}_0^\top$, and
$\boldsymbol{\bar{W}}_0 := \frac{1}{t}\boldsymbol{W}_0 \boldsymbol{D}_0^\top$.
Then
$\boldsymbol{\bar{X}}_1 = \boldsymbol{A}\boldsymbol{\bar{X}}_0 + \boldsymbol{B}\boldsymbol{\bar{U}}_0 + \boldsymbol{\bar{W}}_0$.
Under Assumption~\ref{ass:PEo}, there exists a unique matrix $\boldsymbol{V}\in\mathbb{R}^{(m+n)\times n}$ such that
$\begin{bmatrix} \boldsymbol{K}\\ \boldsymbol{I}_n \end{bmatrix} = \boldsymbol{\Phi}\boldsymbol{V} = \begin{bmatrix} \boldsymbol{\bar{U}}_0\boldsymbol{V}\\ \boldsymbol{\bar{X}}_0\boldsymbol{V} \end{bmatrix}$.
Hence, the closed-loop matrix can be written as
$\boldsymbol{A}+\boldsymbol{B}\boldsymbol{K}
=
(\boldsymbol{\bar{X}}_1-\boldsymbol{\bar{W}}_0)\boldsymbol{V}$.

Since $\{\boldsymbol{d}_k\}$ is uniformly bounded and $\{\boldsymbol{w}_k\}$ is zero-mean with finite covariance, $\|\boldsymbol{\bar{W}}_0\|_2 \xrightarrow[t\to\infty]{\mathrm{a.s.}} 0$. Thus, for large samples, $\boldsymbol{A}+\boldsymbol{B}\boldsymbol{K}\approx \boldsymbol{\bar{X}}_1\boldsymbol{V}$. The covariance estimate is
\begin{equation}\label{estimatedmodelnoise}
\boldsymbol{\hat{\varepsilon}}_0
=
\boldsymbol{X}_1-\boldsymbol{\bar{X}}_1\boldsymbol{V}\boldsymbol{X}_0,\qquad
\boldsymbol{\hat{U}}_\epsilon
=
\frac{1}{t}\boldsymbol{\hat{\varepsilon}}_0\boldsymbol{\hat{\varepsilon}}_0^\top.
\end{equation}

Using $\begin{bmatrix} \boldsymbol{K}\\ \boldsymbol{I}_n \end{bmatrix} = \boldsymbol{\Phi}\boldsymbol{V}$ and $\boldsymbol{K}=\boldsymbol{\bar{U}}_0\boldsymbol{V}$, the CE-LQR problem in~\eqref{eq:ceop_u} can be rewritten in terms of $\boldsymbol{V}$ as
\begin{subequations}\label{diop2}
\begin{align}
\min_{\boldsymbol{V},\boldsymbol{\hat{U}}_K}
J(\boldsymbol{V})
&=
\mathrm{Tr}\!\left(
(\boldsymbol{Q}+\boldsymbol{V}^\top \boldsymbol{\bar{U}}_0^\top \boldsymbol{R}\boldsymbol{\bar{U}}_0 \boldsymbol{V})
\boldsymbol{\hat{U}}_K
\right), \label{diop2a}\\
\text{s.t.}\quad
\boldsymbol{\hat{U}}_K
&=
\boldsymbol{\hat{U}}_\epsilon
+
(\boldsymbol{\bar{X}}_1\boldsymbol{V})\boldsymbol{\hat{U}}_K(\boldsymbol{\bar{X}}_1\boldsymbol{V})^\top, \label{diop2b}\\
\boldsymbol{\bar{X}}_0\boldsymbol{V}
&=
\boldsymbol{I}. \label{diop2c}
\end{align}
\end{subequations}
Let $\boldsymbol{V}:=\boldsymbol{\Phi}^{-1}\begin{bmatrix}\boldsymbol{K}\\ \boldsymbol{I}_n\end{bmatrix}$. Then~\eqref{eq:ceop_u} and~\eqref{diop2} are equivalent, and the optimal gain is recovered as $\boldsymbol{K}^\star=\boldsymbol{\bar U}_0\boldsymbol{V}^\star$~\cite{zhao2025data}.

\subsection{DeePO Implementation}
\begin{assumption}[Feasibility of covariance parametrization]
\label{ass:feasible_set}
The covariance parameterization of the LQR problem in~\eqref{diop2} admits a nonempty feasible set. Specifically, there exist constants $\alpha\in(0,1)$ and $\bar K>0$ such that
$\mathcal{S}
:=
\left\{
\boldsymbol{V} \middle|
\boldsymbol{\bar X}_0 \boldsymbol{V} = \boldsymbol{I}_n,
\rho(\boldsymbol{\bar X}_1 \boldsymbol{V})\le 1-\alpha,
\|\boldsymbol{\bar U}_0 \boldsymbol V\|_2\le \bar K
\right\}$
is nonempty.
\end{assumption}

\begin{theorem}
\label{thm:gradV}
Consider~\eqref{diop2}. Assume that $\boldsymbol{V}$ satisfies $\boldsymbol{\bar X}_0\boldsymbol{V}=\boldsymbol{I}$ and that $\boldsymbol{\hat U}_K$ is the unique solution of~\eqref{diop2b} associated with $\boldsymbol{V}$. Then the gradient of $J(\boldsymbol{V})$ in~\eqref{diop2a} is
\begin{equation}
\nabla_{\boldsymbol{V}} J(\boldsymbol{V})
=
2\left(
\boldsymbol{\bar U}_0^\top \boldsymbol{R}\boldsymbol{\bar U}_0
+
\boldsymbol{\bar X}_1^\top \boldsymbol{P}_V \boldsymbol{\bar X}_1
\right)\boldsymbol{V}\boldsymbol{\hat U}_K.
\end{equation}
\end{theorem}
The proof is deferred to Appendix~\ref{app:proof_gradV}.

\subsubsection{Rank-1 gradient-descent implementation of DeePO }
In the adaptive control setting, we collect online closed-loop data  
$\boldsymbol{(x_t, u_t, x_{t+1})}$ at each time step $t$, which are used to form the data matrices  
$\boldsymbol{(X_{0,t+1}, U_{0,t+1}, X_{1,t+1})}$.  
These data enable a single projected gradient-descent update of the parameterized policy at time~$t$.  The updated policy is then applied to the system, and the procedure is repeated iteratively.  
The method is summarized in \textbf{Algorithm~1}.

\begin{algorithm}[h]
\caption{Rank-1 GD DeePO}
\label{alg:rank1_gd_deepo}
\KwIn{Initial stabilizing policy $\boldsymbol{K}_{t_0}$, stepsize $\eta$, and offline data $(\boldsymbol{X}_{0,t_0},\boldsymbol{U}_{0,t_0},\boldsymbol{X}_{1,t_0})$.}
\For{$t=t_0,t_0+1,\dots$}{
    Apply $\boldsymbol{u}_t=\boldsymbol{K}_t\boldsymbol{x}_t$ and observe $\boldsymbol{x}_{t+1}$;

    Form $\boldsymbol{\phi}_t:=\begin{bmatrix}\boldsymbol{u}_t^\top & \boldsymbol{x}_t^\top\end{bmatrix}^\top$;\\
    Update $\bar{\boldsymbol X}_{0,t+1}$, $\bar{\boldsymbol U}_{0,t+1}$, $\bar{\boldsymbol X}_{1,t+1}$ recursively;\\
    Update $\boldsymbol{\Phi}_{t+1}=\frac{t\boldsymbol{\Phi}_t+\boldsymbol{\phi}_t\boldsymbol{\phi}_t^\top}{t+1}$ and compute $\boldsymbol{\Phi}^{-1}_{t+1}$;\\
    Compute $\boldsymbol{V}_{t+1}=\boldsymbol{\Phi}_{t+1}^{-1}\begin{bmatrix}\boldsymbol{K}_t\\ \boldsymbol{I}_n\end{bmatrix}$;\\
    Perform one-step projected GD:
    \begin{equation}\label{update2}
    \boldsymbol{V}'_{t+1}
    =
    \boldsymbol{V}_{t+1}
    -
    \eta\,
    \boldsymbol{\Pi}_{\bar{\boldsymbol X}_{0,t+1}}
    \nabla J_{t+1}(\boldsymbol{V}_{t+1});
    \end{equation}
    Update the control gain:
    \begin{equation}\label{update3}
    \boldsymbol{K}_{t+1}
    =
    \bar{\boldsymbol U}_{0,t+1}\boldsymbol{V}'_{t+1}.
    \end{equation}
}
\end{algorithm}

\begin{assumption}[Persistency of excitation for online DeePO]
\label{ass:PE}
For all $t \ge t_0$, the input sequence
$\{\boldsymbol{u}_k\}_{k=0}^{t-1}$ provides a uniform level of excitation.
Specifically, the input matrix
$\boldsymbol{U}_{0,t} := [\boldsymbol{u}_0,\dots,\boldsymbol{u}_{t-1}]$
is persistently exciting of order $n+1$ with excitation level $\gamma\sqrt{t(n+1)}$,
in the sense that the associated Hankel matrix
\[
\mathcal H_{n+1}(\boldsymbol{U}_{0,t})
=
\begin{bmatrix}
\boldsymbol{u}_0   & \boldsymbol{u}_1    & \cdots & \boldsymbol{u}_{t-n-1} \\
\boldsymbol{u}_1   & \boldsymbol{u}_2    & \cdots & \boldsymbol{u}_{t-n}   \\
\vdots             & \vdots              & \ddots & \vdots                 \\
\boldsymbol{u}_n   & \boldsymbol{u}_{n+1}& \cdots & \boldsymbol{u}_{t-1}
\end{bmatrix}
\]
satisfies
$\sigma_{\min}\!\bigl(\mathcal H_{n+1}(\boldsymbol{U}_{0,t})\bigr)
\ge
\gamma\sqrt{t(n+1)}$,
for some constant $\gamma>0$ independent of $t$.
\end{assumption}

\begin{remark}
Assumption~\ref{ass:PE} ensures that sufficiently informative data are available for covariance estimation and online policy updates, and together with Assumption~\ref{ass:noise_bounded} determines the SNR used in the subsequent analysis~\cite{zhao2025data}.
\end{remark}
The signal-to-noise ratio (SNR) is defined as $\gamma/\delta$, where $\gamma$ and $\delta$ are defined in Assumptions~\ref{ass:PE} and~\ref{ass:noise_bounded}, respectively.

The sample covariance matrices are updated recursively. Let $\boldsymbol{\phi}_t:=\begin{bmatrix}\boldsymbol{u}_t^\top & \boldsymbol{x}_t^\top\end{bmatrix}^\top$. Then, for example,
$\bar{\boldsymbol X}_{0,t+1}=\frac{t}{t+1}\bar{\boldsymbol X}_{0,t}+\frac{1}{t+1}\boldsymbol{x}_t\boldsymbol{\phi}_t^\top$,
and similarly for $\bar{\boldsymbol U}_{0,t+1}$ and $\bar{\boldsymbol X}_{1,t+1}$. The sample covariance matrix satisfies
$\boldsymbol{\Phi}_{t+1} = \frac{t \boldsymbol{\Phi}_t + \boldsymbol{\phi}_t \boldsymbol{\phi}_t^\top}{t+1}$.
By the Sherman--Morrison formula~\cite{sherman1950adjustment}, $\boldsymbol{\Phi}^{-1}_{t+1}$ is calculated as:
$\boldsymbol{\Phi}^{-1}_{t+1} = \frac{t+1}{t}\left(\boldsymbol{\Phi}_t^{-1} - \frac{\boldsymbol{\Phi}_t^{-1} \boldsymbol{\phi}_t \boldsymbol{\phi}_t^\top \boldsymbol{\Phi}_t^{-1}}{t + \boldsymbol{\phi}_t^\top \boldsymbol{\Phi}_t^{-1} \boldsymbol{\phi}_t}\right)$.

The projection operator $\boldsymbol{\Pi}_{\bar{\boldsymbol X}_{0,t+1}}$ in~\eqref{update2} denotes the orthogonal projection onto the tangent space of the affine constraint $\bar{\boldsymbol X}_{0,t+1}\boldsymbol V=\boldsymbol I$. It preserves the equality constraint to first order (\ref{diop2c}) and can be implemented either by the closed-form projector $\boldsymbol I-\bar{\boldsymbol X}_{0,t+1}^\top(\bar{\boldsymbol X}_{0,t+1}\bar{\boldsymbol X}_{0,t+1}^\top)^{-1}\bar{\boldsymbol X}_{0,t+1}$ or via a null-space parameterization.

For any gain $\boldsymbol K$, let $C(\boldsymbol K)$ denote the true steady-state cost and $\hat C(\boldsymbol K)$ its covariance-estimated counterpart. In the lifted formulation, let $J(\boldsymbol V)$ denote the corresponding objective in~\eqref{diop2}. We use the optimality gap $C(\boldsymbol{K}_t)-C^\star$ to quantify the convergence of the controller. Accordingly, define the regret
\begin{equation}\label{eq:regret}
\mathrm{Regret}_T
:=
\frac{1}{T}
\sum_{t=t_0}^{t_0+T-1}
\bigl(C(\boldsymbol{K}_t)-C^\star\bigr).
\end{equation}
This regret measures the average steady-state performance gap between the online controller and the optimal steady-state policy.  Let $\boldsymbol{e}_{\mathrm{noise}}:=\boldsymbol{U}_\epsilon-\widehat{\boldsymbol{U}}_\epsilon$ denote the covariance estimation error induced by finite data. The optimality gap ~\eqref{eq:regret} can be decomposed as
\small{\begin{equation}
\begin{aligned}\label{comp}
C(\boldsymbol{K}_t)-C^\star
=
\underbrace{
\bigl(C(\boldsymbol{K}_t)-\widehat{C}(\boldsymbol{K}_t)\bigr)
}_{\text{noise mismatch}}
+
\underbrace{
\bigl(\widehat{C}(\boldsymbol{K}_t)-\widehat{C}^\star\bigr)
}_{\text{CE regret}}
+
\underbrace{
\bigl(\widehat{C}^\star-C^\star\bigr)
}_{\text{optimal bias}},
\end{aligned}
\end{equation}}

\noindent where $\text{noise mismatch}=C(\boldsymbol K_t)-\hat C(\boldsymbol K_t)
=
\mathrm{Tr}\!\left(
\boldsymbol{P}_{K_t}\,
\boldsymbol{e_{\mathrm{noise}}}
\right)$ and $\text{optimal bias}=\widehat{C}^\star-C^\star=\mathrm{Tr}(\boldsymbol P_{K^\star}\,\boldsymbol e_{\mathrm{noise}})$.

\begin{lemma}\label{ceregret}
Under Assumptions~\ref{ass:aug_cont}--\ref{ass:PE}, let Algorithm~\ref{alg:rank1_gd_deepo} run for
$T:=t-t_0+1$ steps,   
there exist positive constants $c_i>0$, $i\in\{1,2,3,4\}$,
depending on $(\bar x,\bar u,\|\boldsymbol{R}\|_2,\sigma(\boldsymbol{Q}),$ $\sigma(\boldsymbol{R}),\hat{C}^\star)$,
such that, if the stepsize satisfies $\eta\in(0,c_1]$ and
$\mathrm{SNR}\ge c_2$, then
\begin{equation}\label{ceregretori}
\frac{1}{T}
\sum_{t=t_0}^{t_0+T-1}\text{CE regret}
\;\le\;
\frac{c_3}{\sqrt{T}}
\;+\;
c_4\,\mathrm{SNR}^{-1/2}.
\end{equation}
\end{lemma}
The proof follows the framework of Theorem~2 in~\cite{zhao2025data}, which relies on projected gradient dominance and local smoothness of the objective over the feasible set. In our setting, replacing the unit disturbance covariance with a bounded covariance estimate $\hat{\boldsymbol{U}}_\epsilon$ does not change the proof structure; it only rescales the associated constants. Since our focus is on DHS modeling and covariance-aware robustification, the full derivation is omitted for brevity.
 
\begin{lemma}
\label{lem:PK_bound}
Under Assumption~\ref{ass:feasible_set}, there exists a constant $P_{\max,2}>0$ such that
for all $\boldsymbol V\in\mathcal S$ with $\boldsymbol K=\bar{\boldsymbol U}_0\boldsymbol V$, the Lyapunov solution $\boldsymbol P_{\boldsymbol K}$ in (\ref{lyapunovdefine}) satisfies
$\|\boldsymbol P_{\boldsymbol K}\|_2\le P_{\max,2}$.
Consequently, $\|\boldsymbol P_{\boldsymbol K}\|_\star
\le
n\,\|\boldsymbol P_{\boldsymbol K}\|_2
\le
n P_{\max,2}
=:P_{\max}$,
uniformly for all $\boldsymbol V\in\mathcal S$.
\end{lemma}
The proof is deferred to Appendix~\ref{app:PK_bound}.

\begin{assumption}\textbf{(Covariance estimation mismatch)}
\label{ass:cov_mismatch}
We assume that the mismatch $e_{\mathrm{noise}}$ is bounded in $L_2$ norm: $\|e_{\mathrm{noise}}\|_2 \le \bar\delta_{\mathrm{noise}}$.
\end{assumption}

\begin{theorem}
\label{thm:SGD}
Suppose 
Assumptions~\ref{ass:aug_cont}--\ref{ass:cov_mismatch}
hold, and let Algorithm~\ref{alg:rank1_gd_deepo} run for 
$T:=t-t_0+1$ iterations given offline data $(\boldsymbol{X}_{0,t_0}, \boldsymbol{U}_{0,t_0},$ $ \boldsymbol{X}_{1,t_0})$,
there exist positive constants $c_i>0$, $i\in\{1,2,3,4\}$,
depending on $(\bar x,\bar u,\|\boldsymbol{R}\|_2,\sigma(\boldsymbol{Q}),$ $\sigma(\boldsymbol{R}),C^\star)$,
such that, if the stepsize satisfies $\eta\in(0,c_1]$ and
$\mathrm{SNR}\ge c_2$,  
the regret satisfies
\begin{equation}\label{eq:regret-final}
\mathrm{Regret}_T\le\frac{c_3}{\sqrt{T}}
+
c_4\,\mathrm{SNR}^{-1/2}+2P_{\max}\bar\delta_{\mathrm{noise}}.
\end{equation}
\end{theorem}
The proof is deferred to Appendix~\ref{app:thm:SGD}.

\subsubsection{ADAM–GD Implementation of DeePO}
To improve the performance of online policy updates, we replace the lifted-variable standard gradient step in~\eqref{update2} with an ADAM-style preconditioned update~\eqref{eq:adam}--\eqref{eq:proj}, while keeping the same online recursive covariance updates as in DeePO. The resulting ADAM--DeePO procedure is summarized in Algorithm~\ref{alg:adam}. Here, the ADAM moments $m_t,v_t,\hat m_t,\hat v_t$ have the same dimension as $\boldsymbol V_t$, namely $\mathbb R^{(m+n)\times n}$. In~\eqref{eq:adam_step}, all nonlinearities act elementwise: $\sqrt{\hat v_{t+1}}$ is taken entrywise, and $\hat m_{t+1}/(\sqrt{\hat v_{t+1}}+\epsilon)$ denotes elementwise division. Accordingly, $\boldsymbol D_{t+1}:=\sqrt{\hat{\boldsymbol v}_{t+1}}+\epsilon$ is used only for elementwise scaling.

\begin{algorithm}[htbp]
\caption{ADAM--DeePO for Direct Adaptive LQR Policy Learning}
\label{alg:adam}
\KwIn{Initial stabilizing gain $\boldsymbol K_{t_0}$, $\boldsymbol m_{t_0}=\mathbf 0$, $\boldsymbol v_{t_0}=\mathbf 0$; initial stepsize $\eta_0>0$ (and a stepsize schedule $\{\eta_t\}_{t\ge t_0}$); ADAM hyperparameters $\beta_1,\beta_2\in(0,1)$ and $\epsilon>0$; offline data $(X_{0,t_0},U_{0,t_0},X_{1,t_0})$.}

\For{$t=t_0,t_0+1,\dots$}{
    Apply control $u_t=\boldsymbol K_t x_t$ and observe $x_{t+1}$;

    Form $\boldsymbol{\phi}_t:=\begin{bmatrix}\boldsymbol{u}_t^\top & \boldsymbol{x}_t^\top\end{bmatrix}^\top$;\\
    Update $\bar{\boldsymbol X}_{0,t+1}$, $\bar{\boldsymbol U}_{0,t+1}$, $\bar{\boldsymbol X}_{1,t+1}$ recursively;\\
    Update $\boldsymbol{\Phi}_{t+1}=\frac{t\boldsymbol{\Phi}_t+\boldsymbol{\phi}_t \boldsymbol{\phi}_t^\top}{t+1}$
    and compute $\boldsymbol{\Phi}^{-1}_{t+1}$.

    Given $\boldsymbol K_t$, compute lifted representation:
    \[
    \boldsymbol V_{t+1}
    =
    \boldsymbol \Phi_{t+1}^{-1}
    \begin{bmatrix}
        \boldsymbol K_t\\[2pt]
        \boldsymbol I
    \end{bmatrix}.
    \]

    Compute stochastic gradient $\boldsymbol g_{t+1}:=\nabla J_{t+1}(\boldsymbol V_{t+1})$;

    \textbf{ADAM moment updates:}
    \begin{subequations}\label{eq:adam}
    \begin{align}
        \boldsymbol m_{t+1}
        =&~ \beta_1 \boldsymbol m_t + (1-\beta_1)\boldsymbol g_{t+1}, \label{eq:adam_m}\\
        \boldsymbol v_{t+1}
        =&~ \beta_2 \boldsymbol v_t + (1-\beta_2)\big(\boldsymbol g_{t+1}\odot \boldsymbol g_{t+1}\big), \label{eq:adam_v}\\
        \hat{\boldsymbol m}_{t+1}
        =&~ \boldsymbol m_{t+1}/\big(1-\beta_1^{\,t-t_0+1}\big), \label{eq:adam_mhat}\\
        \hat{\boldsymbol v}_{t+1}
        =&~ \boldsymbol v_{t+1}/\big(1-\beta_2^{\,t-t_0+1}\big), \label{eq:adam_vhat}\\
        \tilde{\boldsymbol V}_{t+1}
        =&~ \boldsymbol V_{t+1}
           - \eta_t\,
             \frac{\hat{\boldsymbol m}_{t+1}}{\sqrt{\hat{\boldsymbol v}_{t+1}} + \epsilon} \notag\\
        =&~ \boldsymbol V_{t+1}
           - \tilde{\eta}_t\,
             \frac{\boldsymbol m_{t+1}}{\boldsymbol D_{t+1}}, \label{eq:adam_step}
    \end{align}
    \end{subequations}
    where $\tilde{\eta}_t:=\eta_t/\big(1-\beta_1^{\,t-t_0+1}\big)$ and
    $\boldsymbol D_{t+1}:=\sqrt{\hat{\boldsymbol v}_{t+1}}+\epsilon$ (elementwise).

    \textbf{Affine projection to satisfy $\bar{\boldsymbol X}_{0,t+1}\boldsymbol V=\boldsymbol I$:}
    \begin{equation}\label{eq:proj}
    \begin{aligned}
        \boldsymbol V'_{t+1}
        =
        &\tilde{\boldsymbol V}_{t+1}
        +
        \bar{\boldsymbol X}_{0,t+1}^\top
        \big(\bar{\boldsymbol X}_{0,t+1}\bar{\boldsymbol X}_{0,t+1}^\top\big)^{-1}\\
        &\Big(
            \boldsymbol I
            -
            \bar{\boldsymbol X}_{0,t+1}\tilde{\boldsymbol V}_{t+1}
        \Big).
    \end{aligned}
    \end{equation}

    Update control gain:
    \[
        \boldsymbol K_{t+1}=\bar{\boldsymbol U}_{0,t+1}\boldsymbol V'_{t+1}.
    \]
}
\end{algorithm}

\begin{assumption}
\label{ass:adam-pre}
Let $\boldsymbol D_t := \sqrt{\hat{\boldsymbol v}_t}+\epsilon$
denote the elementwise ADAM scaling matrix. There exist constants
$0<c_{\min}\le c_{\max}<\infty$ such that, for all entries $(i,j)$ and all $t$,
$c_{\min}\le (\boldsymbol D_t)_{ij}^{-1}\le c_{\max}$.
\end{assumption}

\begin{assumption}[Effective stepsize schedule]
\label{ass:adam-steps}
Let
$\tilde{\eta}_t := \eta_t/$

\noindent$(1-\beta_1^{t+1})$
denote the bias-corrected effective stepsize.
Assume that $\{\tilde{\eta}_t\}$ is positive, nonincreasing, and uniformly bounded,
i.e.,
$0 < \tilde{\eta}_t \le \tilde{\eta}_{\max}$ for all $t$.
Moreover, the stepsizes are chosen such that the cumulative stepsize grows
sublinearly with the time horizon, satisfying
$r_1 \sqrt{T}
\;\le\;
\sum_{t=t_0}^{t_0+T-1} \tilde{\eta}_t
\;\le\;
r_2 \sqrt{T}$,
for some constants $r_1,r_2>0$ and all $T\ge1$.
\end{assumption}

\begin{theorem}
\label{thm:adam-regret}
Under Assumptions~\ref{ass:aug_cont}--\ref{ass:adam-steps}, let Algorithm~\ref{alg:adam} run for $T:=t-t_0+1$ steps. Then there exist constants $d_1,d_2,d_3,d_4>0$, depending on $(\bar x,\bar u,\|\boldsymbol{R}\|_2,\sigma(\boldsymbol{Q}),\sigma(\boldsymbol{R}),\hat{C}^\star)$, such that, if $\tilde{\eta}_{\max}\in(0,d_1]$ and $\mathrm{SNR}\ge d_2$, then
\begin{equation}\label{eq:regret-adam}
\mathrm{Regret}_T
\le
\frac{d_3}{\sqrt{T}}
+
d_4\,\mathrm{SNR}^{-1/2}
+
2P_{\max}\bar{\delta}_{\mathrm{noise}}.
\end{equation}
\end{theorem}

The proof is deferred to Appendix~\ref{app:thm:adam-regret}.

In practice, the proposed update improves online learning performance without sacrificing performance guarantees.
Moreover, Assumption~\ref{ass:noise_bounded} is adopted for analytical clarity; extending the result to light-tailed stochastic disturbances is left for future work.

\section{Simulation}
This section evaluates the proposed method on an industrial-park DHS in Northern China~\cite{yi2023energy} under realistic model mismatch and stochastic disturbances. The DHS study considers both static and time-varying parameter perturbations to emulate practical uncertainty in mass flow rates and heat-transfer characteristics. Overall, the results show that the proposed method achieves stable near-optimal DHS operation and improved robustness under these uncertain operating conditions. Additional mechanism validation on a standard three-dimensional benchmark~\cite{zhao2025data}, illustrating the effects of disturbance-covariance estimation and ADAM-based updates under stochastic excitation, is reported in Appendix~\ref{app:3dbenchmark}.

\subsection{Industrial DHS in Northern China}
We consider an industrial-park DHS in Northern China with three producers and eight loads~\cite{yi2023energy}, yielding an augmented system with $n=22$ states and $m=11$ inputs. The simulations are performed with a sampling interval of $\tau=0.1\,\mathrm{s}$. We apply Gaussian process noise $\boldsymbol{w}_t = 0.0042\,\boldsymbol{v}^{(w)}_t~(\mathrm{kW})$ and exploration noise $\boldsymbol{w}_{s,t} = 0.1\,\boldsymbol{v}^{(s)}_t~(\mathrm{kW})$, where $\boldsymbol{v}^{(w)}_t,\boldsymbol{v}^{(s)}_t \sim \mathcal{N}(\boldsymbol{0},\boldsymbol{I}_{11})$ are independent. The former captures stochastic demand fluctuations and unmodeled thermal effects, while the latter provides excitation for online learning.

We compare GD-DeePO and ADAM-DeePO under model discrepancy by initializing the controller with an available design model rather than the exact current plant. In practice, such a model may come from a model estimated at an earlier operating condition, or an engineering model stored from historical records. Because DHS operating conditions can vary over time, the model available for controller design may differ from the actual closed-loop dynamics. We first consider a static discrepancy setting to isolate the effect of initialization error and online correction in a controlled manner. The real system is $\boldsymbol{A}_{\mathrm{real}}=\begin{bmatrix}\boldsymbol{I}-\tau\boldsymbol{A}_1 & \boldsymbol{0}\\ \boldsymbol{C}_T & \boldsymbol{I}\end{bmatrix}$, whereas the model used to compute the initial feedback gain $\boldsymbol{K}_0$ is $\boldsymbol{A}_{\mathrm{init}}=\begin{bmatrix}\boldsymbol{I}-(1+\delta_e)\tau\boldsymbol{A}_1 & \boldsymbol{0}\\ \boldsymbol{C}_T & \boldsymbol{I}\end{bmatrix}$. Here, $\delta_e$ quantifies the discrepancy in the thermal time scale between the design model and the actual plant. A positive $\delta_e$ means that the design model assumes faster thermal dynamics than those of the plant, whereas a negative $\delta_e$ corresponds to slower assumed dynamics. The controller is initialized using the LQR solution of $\boldsymbol{A}_{\mathrm{init}}$, after which DeePO updates the feedback gain $\boldsymbol{K}$ online from closed-loop data generated by $\boldsymbol{A}_{\mathrm{real}}$.

Beyond this static discrepancy setting, we also consider time-varying model perturbations. Specifically, the thermal time-scale perturbation $\delta_e$ is augmented by a zero-mean bounded time-varying component $\delta_e^c(t)$, yielding the effective perturbation $\delta_e(1+\delta_e^c(t))$. This emulates variations in operating conditions such as changes in mass flow rates and heat-transfer characteristics, and results in a linear time-varying system whose dynamics deviate persistently from the design model used to compute $\boldsymbol{K}_0$.

\subsection{Near-optimal and stable DHS operation}
We next illustrate the closed-loop operation of ADAM-DeePO under a $40\%$ static model discrepancy and a constant heat demand of $10\,\mathrm{MW}$. Figures~\ref{fig:T_traj}--\ref{fig:hG_traj} show a representative run. The representative trajectories show that the proposed controller maintains well-behaved closed-loop operation throughout the process. The temperature and heat-generation trajectories remain bounded and approach steady operating levels, while the optimality error decreases rapidly after the initial transient and remains close to zero during online learning. These results show that ADAM-DeePO can improve the feedback policy online while maintaining stable DHS operation under static model discrepancy.

\vspace{-0.2cm}
\begin{center}
\hspace{-0.2cm}
    \includegraphics[width=1.0\linewidth]{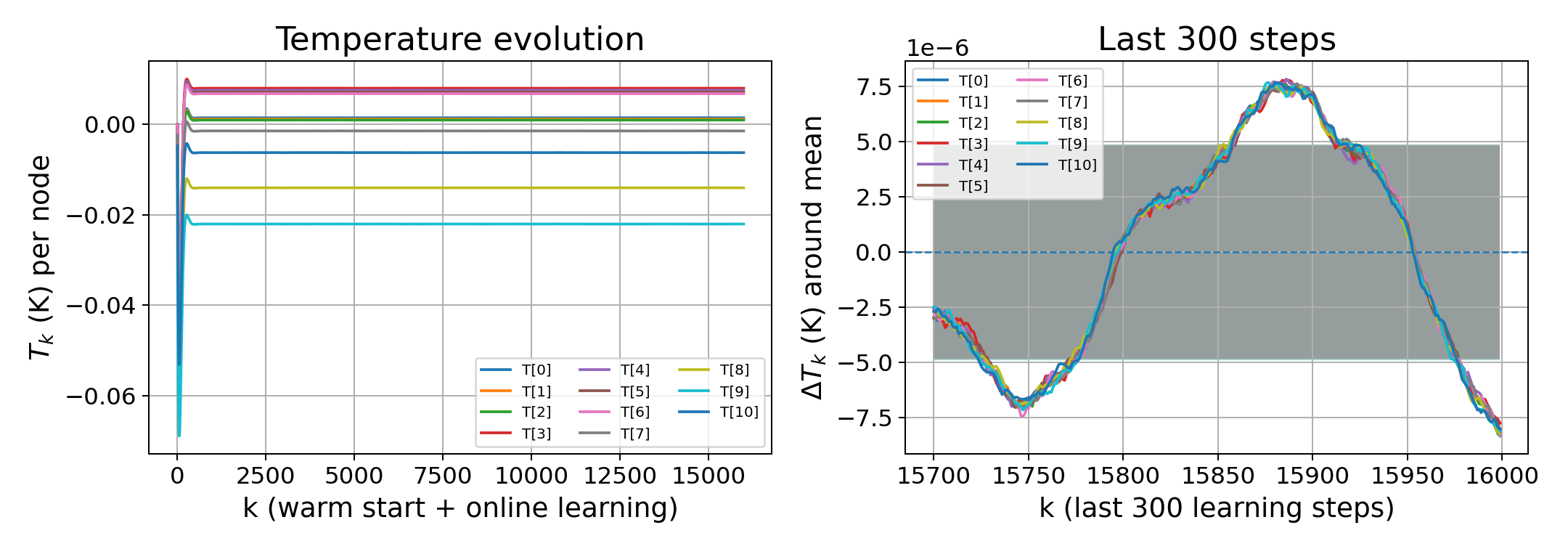}
    \vspace{-0.4cm}
    \captionof{figure}{Temperature evolution under ADAM-DeePO.}
    \label{fig:T_traj}
\end{center}
\vspace{-0.7cm}
\begin{center}
    \includegraphics[width=0.9\linewidth]{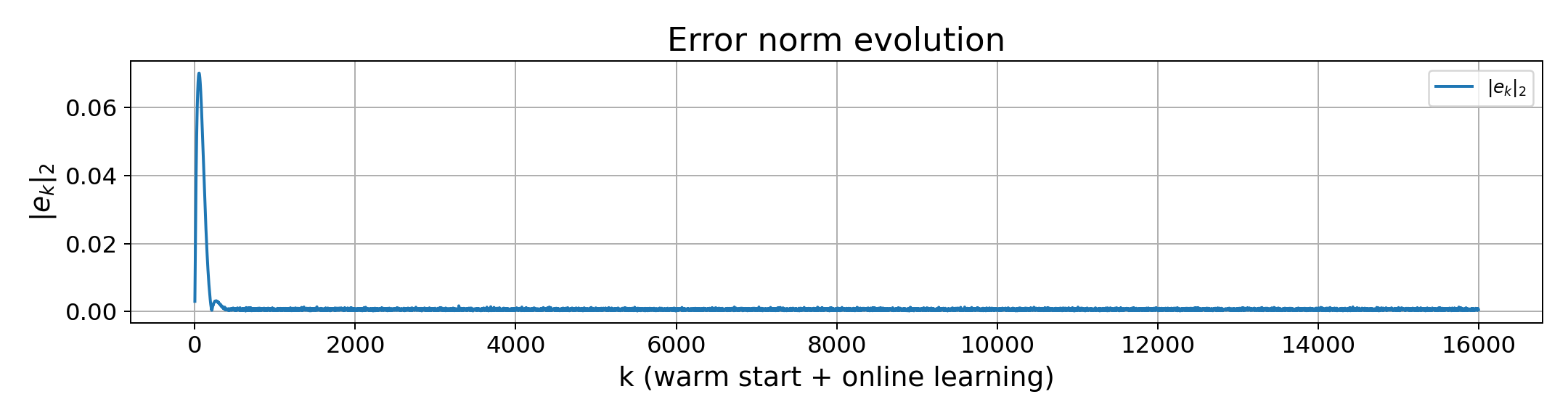}
    \vspace{-0.5cm}
    \captionof{figure}{Optimality error evolution under ADAM-DeePO.}
    \label{fig:e_traj}
\end{center}

\vspace{-0.5cm}
\begin{center}
    \includegraphics[width=1.0\linewidth]{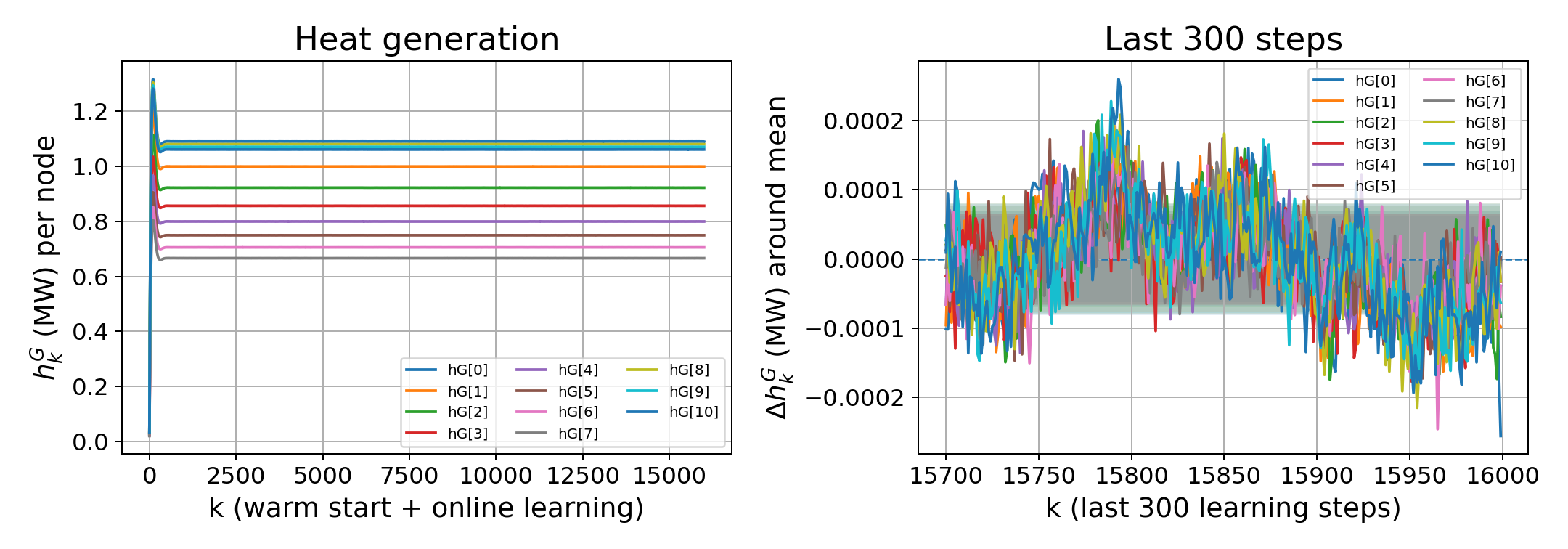}
    \vspace{-0.6cm}
    \captionof{figure}{Heat generation under ADAM-DeePO.}
    \label{fig:hG_traj}
\end{center}
\vspace{-0.2cm}

To assess reproducibility, we repeated the same experiment over $5$ random seeds. The relative cost error is substantially reduced from its initial value $\frac{|C(\boldsymbol{K}_0)-C^\star|}{C^\star}=1.697\times10^{-2}$ to a final value $\frac{|C(\boldsymbol{K}_t)-C^\star|}{C^\star}$ with mean $9.606\times10^{-4}$ and standard deviation $8.531\times10^{-4}$ across seeds. Moreover, the final optimality error and control increment remain small, with mean $\|e_k\|_2=4.923\times10^{-4}$ and mean $\|u_k\|_2=3.354\times10^{-4}$ over the $5$ runs. These results indicate that ADAM-DeePO can reliably learn a near-optimal controller online and maintain stable and economically efficient DHS operation under model discrepancy.

\subsection{Comparison of ADAM-DeePO and nominal MPC}
We compare ADAM-DeePO with a nominal MPC baseline under the same $40\%$ static model discrepancy and the same stochastic heat-demand disturbance, using the same representative seed as in Figures~\ref{fig:T_traj}--\ref{fig:hG_traj}. The MPC controller is also designed from the augmented system \eqref{augmentd}, and therefore uses the same error state $e_k$ and targets the same economic objective. In this single-hour constant-load setting, MPC is implemented in receding-horizon fashion around a constant heat demand of $10\,\mathrm{MW}$.

Figures~\ref{fig:hg_mpc} and~\ref{fig:error_mpc} show that MPC yields bounded heat-generation trajectories and drives the optimality error to a small neighborhood of the desired steady state, confirming the usefulness of the design in \eqref{augmentd}. However, because the receding-horizon optimization relies on a fixed discrepant prediction model, the error cannot be fully removed under model discrepancy. By contrast, ADAM-DeePO updates the feedback policy directly from closed-loop data and can therefore progressively compensate for the discrepancy.

\vspace{-0.3cm}
\begin{center}
    \includegraphics[width=1.0\linewidth]{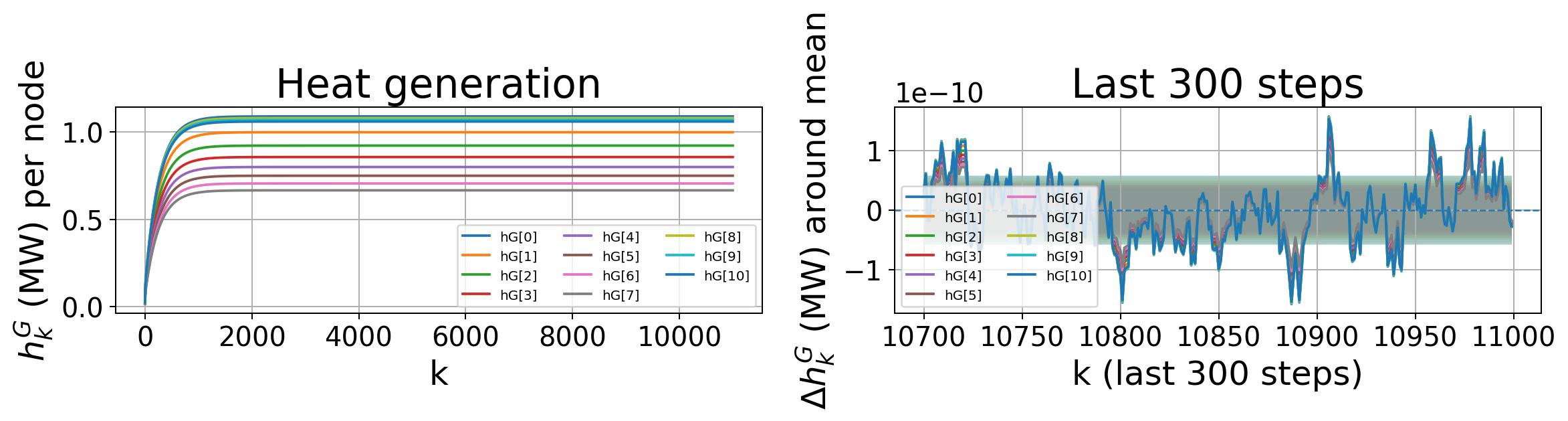}
    \vspace{-0.6cm}
    \captionof{figure}{Heat generation trajectories under nominal MPC.}
    \label{fig:hg_mpc}
\end{center}

\vspace{-0.5cm}
\begin{center}
    \includegraphics[width=1.0\linewidth]{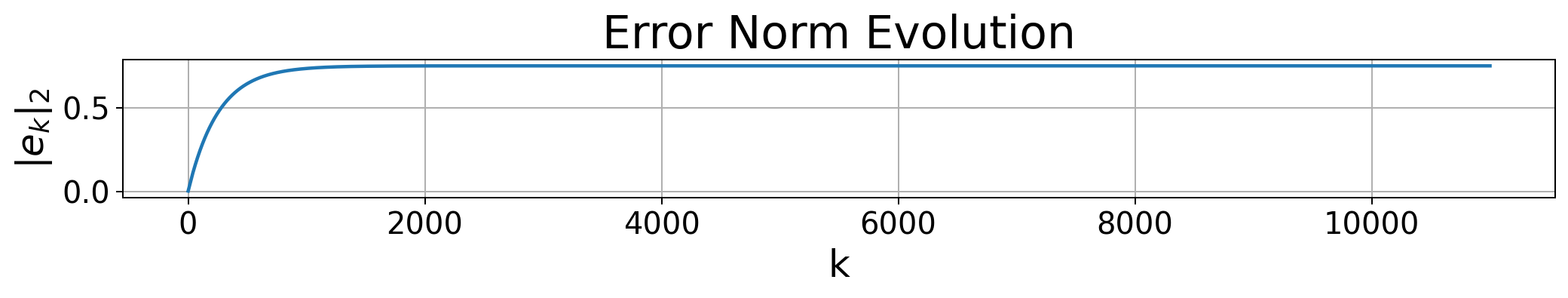}
    \vspace{-0.8cm}
    \captionof{figure}{Optimality error evolution under nominal MPC.}
    \label{fig:error_mpc}
\end{center}
\vspace{-0.2cm}

We further compare the two methods using steady operating metrics aligned with \textbf{E1}--\textbf{E2}, where the steady average generation and temperature costs are evaluated by \eqref{5a} and \eqref{4a}, respectively, over a terminal averaging window. Since both controllers are based on the augmented formulation, the main difference is whether the control law can adapt online to model discrepancy. Under this metric, ADAM-DeePO and nominal MPC achieve almost identical steady average generation cost, indicating similar aggregate heat-supply performance, but their temperature-regulation performance differs substantially. Both methods also satisfy the final heat-balance condition accurately, with balance gaps below $10^{-5}\,\mathrm{MW}$ in magnitude, whereas ADAM-DeePO reaches the economic neighborhood much faster: $37.7\,\mathrm{s}$ versus $1069.8\,\mathrm{s}$ for MPC. These results are consistent with the trajectory plots and show the benefit of online policy adaptation under model discrepancy.

\subsection{Operation under a realistic time-varying heat-demand profile}
We evaluate the learned controller under a realistic time-varying heat disturbance constructed from a real-world hourly heat-demand profile~\cite{yi2023energy}. The full simulation is carried out on the original physical timescale with $\tau=0.1\,\mathrm{s}$ over 24 physical hours. In each hour, the simulation consists of a warm-start phase of $1300$ steps, an online learning phase of $10000$ steps, and a deployment phase of $24700$ steps, giving $36000$ steps per hour in total. This hourly procedure is repeated across the 24-hour demand profile. For visual clarity, Figure~\ref{fig:realload_gen_demand} displays only the deployment phase, with the horizontal axis shown in compressed form while preserving the underlying daily variation pattern.

\vspace{-0.3cm}
\begin{center}
    \includegraphics[width=1.0\linewidth]{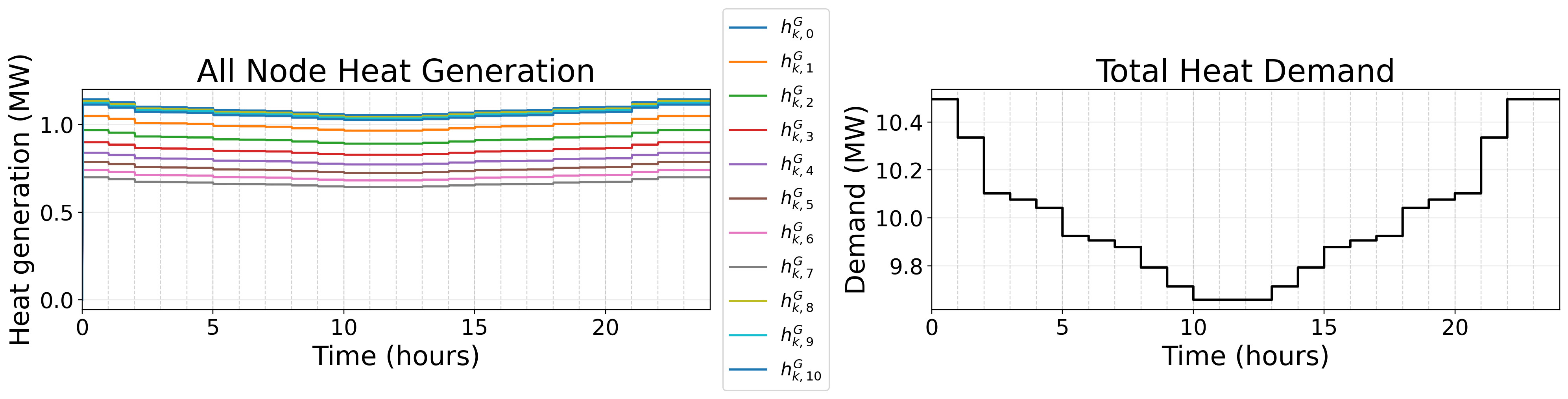}
    \vspace{-0.8cm}
    \captionof{figure}{Heat-generation trajectories under a realistic $24$-hour heat-demand profile (left), and the corresponding total heat-demand profile (right).}
    \label{fig:realload_gen_demand}
\end{center}
\vspace{-0.6cm}
\begin{center}
    \includegraphics[width=0.9\linewidth]{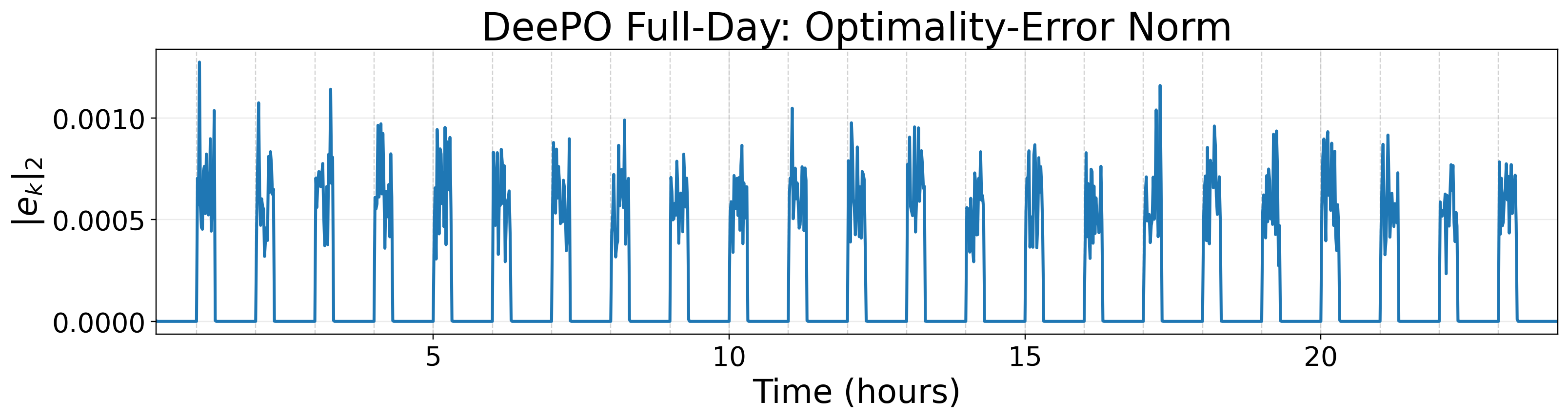}
    \vspace{-0.4cm}
    \captionof{figure}{Evolution of $\|e_k\|_2$.}
    \label{fig:realload_error}
\end{center}
\vspace{-0.2cm}
Figure~\ref{fig:realload_gen_demand} shows that the heat-generation inputs vary smoothly with the demand profile during deployment. Figure~\ref{fig:realload_error} shows the evolution of $\|e_k\|_2$ over the same period. Although the optimality error increases briefly when the hourly demand level changes, it quickly returns to a small neighborhood of zero after each transition. This indicates that the learned controller remains stable and economically consistent under realistic time-varying heat demand.

\subsection{Comparison of ADAM-DeePO and GD-DeePO}
Tables~\ref{modelvariation} and~\ref{modelvariation2} report the relative operating-cost error over five random seeds under different levels of static model discrepancy, using the sample mean and standard deviation (std). The improvement ratio is computed as $\mathrm{IMP}=(\mathrm{mean}_{\mathrm{GD}}-\mathrm{mean}_{\mathrm{ADAM}})/\mathrm{mean}_{\mathrm{GD}}\times 100\%$. Both GD-DeePO and ADAM-DeePO converge to near-optimal solutions, but ADAM-DeePO consistently achieves a lower mean relative cost error with comparable or slightly smaller variability across seeds. These results show that adaptive moment-based scaling improves online learning and closed-loop performance under model discrepancy.

\subsection{Robustness to time-varying thermal dynamics}
Table~\ref{modelvariationtimevary_40} reports the relative operating-cost error over five random seeds under time-varying thermal-model perturbations with a fixed static model discrepancy of $\delta_e=40\%$. For each setting, we report the sample mean and standard deviation (std), and compute the improvement ratio (IMP) from the mean errors of GD-DeePO and ADAM-DeePO. ADAM-DeePO consistently achieves a lower mean relative cost error than GD-DeePO across all tested perturbation levels. As the perturbation magnitude increases, the relative cost error of both methods increases substantially, which is consistent with the fact that larger time-varying perturbations drive the real dynamics further away from the design model.

\begin{table}[htbp]
\centering
\caption{Operating-cost error under static moderate discrepancy.}
\label{modelvariation}
\begin{tabular}{p{1.8cm}|p{1.2cm}|p{1.2cm}|p{1.2cm}|p{1.2cm}}
\hline
$\delta_e$ & \textbf{-40\%} & \textbf{-20\%} & \textbf{20\%} & \textbf{40\%} \\
\hline
\textbf{GD mean}   & 6.256e-2 & 3.762e-2 & 1.347e-2 & 6.732e-3 \\
\textbf{GD std}    & 1.950e-3 & 1.578e-3 & 1.272e-3 & 1.213e-3 \\
\textbf{ADAM mean} & 5.490e-2 & 3.011e-2 & 6.114e-3 & 9.606e-4 \\
\textbf{ADAM std}  & 1.840e-3 & 1.501e-3 & 1.248e-3 & 8.531e-4 \\
\textbf{IMP (mean)}       & 12.25\%  & 19.96\%  & 54.61\%  & 85.73\% \\
\hline
\end{tabular}
\end{table}

\begin{table}[htbp]
\centering
\caption{Operating-cost error under static mild discrepancy.}
\label{modelvariation2}
\begin{tabular}{p{1.8cm}|p{1.2cm}|p{1.2cm}|p{1.2cm}|p{1.2cm}}
\hline
$\delta_e$ & \textbf{-1\%} & \textbf{1\%} & \textbf{-2\%} & \textbf{2\%} \\
\hline
\textbf{GD mean}   & 2.361e-2 & 2.245e-2 & 2.421e-2 & 2.189e-2 \\
\textbf{GD std}    & 1.387e-3 & 1.373e-3 & 1.395e-3 & 1.366e-3 \\
\textbf{ADAM mean} & 1.619e-2 & 1.504e-2 & 1.678e-2 & 1.448e-2 \\
\textbf{ADAM std}  & 1.337e-3 & 1.326e-3 & 1.344e-3 & 1.320e-3 \\
\textbf{IMP (mean)}       & 31.43\%  & 33.01\%  & 30.67\%  & 33.84\% \\
\hline
\end{tabular}
\end{table}

\begin{table}[htbp]
\centering
\caption{Relative operating-cost error under time-varying model perturbations with fixed static model discrepancy $\delta_e=40\%$.}
\label{modelvariationtimevary_40}
\begin{tabular}{p{1.8cm}|p{1.3cm}|p{1.3cm}|p{1.3cm}}
\hline
$\delta_e^c$ & \textbf{10\%} & \textbf{30\%} & \textbf{50\%} \\
\hline
\textbf{GD mean}   & 3.890e-3 & 1.607e-2 & 4.864e-2 \\
\textbf{GD std}    & 2.245e-3 & 6.745e-3 & 1.695e-2 \\
\textbf{ADAM mean} & 3.607e-3 & 1.459e-2 & 4.716e-2 \\
\textbf{ADAM std}  & 2.980e-3 & 6.702e-3 & 1.694e-2 \\
\textbf{IMP}       & 7.26\%   & 9.21\%   & 3.04\% \\
\hline
\end{tabular}
\end{table}

\section{Conclusion}
This paper develops a data-driven online control framework for DHSs by embedding steady-state economic optimality conditions into the system dynamics. Based on DeePO, the resulting controller enables online learning of near-optimal regulation policies under stochastic disturbances and model mismatch, while providing convergence to an optimal control policy, together with closed-loop stability guarantees. An ADAM-enhanced variant is further developed to improve the performance of online policy updates. Simulations on an industrial-scale DHS show that the proposed method achieves stable near-optimal operation and strong empirical robustness to both static and time-varying model perturbations under different heat-demand disturbance conditions. These results suggest that the proposed data-driven online control is a promising approach for practical DHS operation, especially in large-scale systems where accurate models are difficult to maintain, disturbance forecasts are unreliable, and strong closed-loop guarantees are desired.

\section*{CRediT authorship contribution statement}
\textbf{Xinyi Yi:} Conceptualization, Methodology, Software, Validation, Formal analysis, Visualization, Writing -- original draft, Writing -- review \& editing. 
\textbf{Ioannis Lestas:} Supervision, Formal analysis, Writing -- review \& editing.
\section*{Declaration of competing interest}
The authors declare that they have no known competing financial interests or personal relationships that could have appeared to influence the work reported in this paper.

\appendix
\section{Proof of Theorem \ref{optimalitycondition}}\label{app:optimalitycondition}
The Lagrangian for $\textbf{E1}$ is given by:
$L = \frac{1}{2} \boldsymbol{h^G}^{\top} \boldsymbol{F^G} \boldsymbol{h^G}$

\noindent$ + \boldsymbol{\mu}^{\top} (\boldsymbol{-A_1 T}+ \boldsymbol{B_1}\boldsymbol{h^G}-\boldsymbol{B_2\bar{h}^L})$, where $\boldsymbol{\mu}$ is the dual variable associated with \eqref{5b}.
The KKT conditions yield:
$\frac{\partial L}{\partial \boldsymbol{h^G}} = \boldsymbol{F^G h^G} + \boldsymbol{B_1^\top\mu} = \boldsymbol{0},
\frac{\partial L}{\partial \boldsymbol{T}} = \boldsymbol{A_1}^{\top} \boldsymbol{\mu} = \boldsymbol{0}$, where $\boldsymbol{\mu}=z\boldsymbol{1}$. The conditions can be rewritten as $\boldsymbol{F^G {h^G}^\star = }z\boldsymbol{1}$ such that $F^G_i {h^G_{i}}^\star - F^G_j{h^G_{j}}^\star = 0$
for any $i, j \in \boldsymbol{\mathcal{G}}$, establishing $\boldsymbol{F^M {h^{G}}^\star} = \boldsymbol{0}$ as the optimality condition for the unique solution of $\boldsymbol{{h^G}^\star}$. The Lagrangian for $\textbf{E2}$ is:
$\mathcal{L}(\boldsymbol{T}, z, \boldsymbol{\lambda}) = \frac{1}{2} \boldsymbol{T^\top F^D T} + \boldsymbol{\lambda}^\top (\boldsymbol{T}- \boldsymbol{A_1^\dagger} (\boldsymbol{B_1} \boldsymbol{{h^G}^\star}  \boldsymbol{-B_2\bar{h}^L}) - z\boldsymbol{1})$.
The KKT conditions for $\textbf{E2}$ are (\ref{4b}) together with: $\frac{\partial \mathcal{L}}{\partial \boldsymbol{T}} = \boldsymbol{F^D T + \lambda = 0}, \frac{\partial \mathcal{L}}{\partial z} = -\boldsymbol{\lambda}^\top \boldsymbol{1} = 0$, which gives $\boldsymbol{1^T F^DT^\star=0}$. The optimality condition (\ref{4b}) ensures that $\boldsymbol{{T^\star}}$ satisfy the optimal condition of \textbf{E1}. Given that $\boldsymbol{F^D}$ is positive definite, $\textbf{E2}$ constitutes a convex optimization problem over $\boldsymbol{T}$, thereby confirming unique $\boldsymbol{{T^\star}}$ and that these optimality conditions are both necessary and sufficient.

\section{Proof of Theorem \ref{thm:gradV}}\label{app:proof_gradV}
The proof follows the same differential argument as Lemma~2 in DeePO~\cite{zhao2025data}. In particular, using
$
J(\boldsymbol{V})
=
\mathrm{Tr}(
(\boldsymbol{Q}+\boldsymbol{V}^\top \boldsymbol{\bar U}_0^\top \boldsymbol{R}\boldsymbol{\bar U}_0 \boldsymbol{V})$

\noindent$\boldsymbol{\hat U}_K)
=
\mathrm{Tr}(\boldsymbol{P}_V\boldsymbol{\hat U}_\epsilon),
$
the same recursive differentiation of $\boldsymbol{P}_V$ as in the DeePO proof yields
$dJ=2\mathrm{Tr}\!\left(\boldsymbol{\hat U}_K \boldsymbol{E}_V^\top d\boldsymbol{V}\right),$
where
$
\boldsymbol{E}_V
=
(\boldsymbol{\bar U}_0^\top \boldsymbol{R}\boldsymbol{\bar U}_0
+
\boldsymbol{\bar X}_1^\top \boldsymbol{P}_V \boldsymbol{\bar X}_1)\boldsymbol{V}.
$
Hence,
$
\nabla_{\boldsymbol{V}} J(\boldsymbol{V})
=
2\boldsymbol{E}_V\boldsymbol{\hat U}_K
=
2(
\boldsymbol{\bar U}_0^\top \boldsymbol{R}\boldsymbol{\bar U}_0$
$+
\boldsymbol{\bar X}_1^\top \boldsymbol{P}_V \boldsymbol{\bar X}_1)\boldsymbol{V}\boldsymbol{\hat U}_K.
$
The linear constraint~\eqref{diop2c} is enforced separately by projection in the algorithmic update and therefore does not appear explicitly in the gradient expression.

\section{Proof of Lemma \ref{lem:PK_bound}} \label{app:PK_bound}
Fix $\boldsymbol V\in\mathcal S$ and let $\boldsymbol A_{\boldsymbol K}:=\bar{\boldsymbol X}_1\boldsymbol V$.
By Assumption~\ref{ass:feasible_set}, $\rho(\boldsymbol A_{\boldsymbol K})\le 1-\alpha$ and
$\|\boldsymbol K\|_2\le \bar K$.
The Lyapunov equation
$\boldsymbol P_{\boldsymbol K}
=
\sum_{i=0}^{\infty}(\boldsymbol A_{\boldsymbol K}^\top)^i
(\boldsymbol Q+\boldsymbol K^\top \boldsymbol R \boldsymbol K)
(\boldsymbol A_{\boldsymbol K})^i$
implies $\|\boldsymbol P_{\boldsymbol K}\|_2
\le
\sum_{i=0}^{\infty}\|\boldsymbol A_{\boldsymbol K}^i\|_2^2
\big(\|\boldsymbol Q\|_2$

\noindent$+\|\boldsymbol K\|_2^2\|\boldsymbol R\|_2\big)
\le
\Big(\sum_{i=0}^{\infty}\|\boldsymbol A_{\boldsymbol K}^i\|_2^2\Big)
\big(\|\boldsymbol Q\|_2+\bar K^2\|\boldsymbol R\|_2\big)$.
Since $\rho(\boldsymbol A_{\boldsymbol K})\le 1-\alpha$ uniformly on $\mathcal S$,
the series $\sum_{i\ge0}\|\boldsymbol A_{\boldsymbol K}^i\|_2^2$ is uniformly bounded, hence
$\|\boldsymbol P_{\boldsymbol K}\|_2\le P_{\max,2}$ for all $\boldsymbol V\in\mathcal S$.

\section{Proof of Theorem \ref{thm:SGD}}\label{app:thm:SGD}
By Assumption~\ref{ass:cov_mismatch}, we have $\|e_{\mathrm{noise}}\|_2\le \bar\delta_{\mathrm{noise}}$.
Moreover, with the definition of $P_{\max}$, by Hölder's inequality for Schatten norms \cite[(1.174)]{watrous2018theory},
it holds that $|C(\boldsymbol K_t)-\hat C(\boldsymbol K_t)|=
\Big|\mathrm{Tr}\!\left(\boldsymbol{P}_{K_t}\, \boldsymbol{e_{\mathrm{noise}}}\right)\Big|
\le
\|\boldsymbol{P}_{K_t}\|_\star\,\|\boldsymbol{e_{\mathrm{noise}}}\|_2
\le
P_{\max}\bar\delta_{\mathrm{noise}}$. 

Similarly, $|\hat C^\star-C^\star|
=
\Big|\mathrm{Tr}\!\left(\boldsymbol{P}_{K^\star} \boldsymbol{e_{\mathrm{noise}}}\right)\Big|
\le
\|\boldsymbol{P}_{K^\star}\|_\star\|\boldsymbol{e_{\mathrm{noise}}}\|_2
\le
P_{\max}\bar\delta_{\mathrm{noise}}$.
Thus, $\frac{1}{T}\sum_{t=t_0}^{t_0+T-1}
(|C(\boldsymbol{K_t})-\hat C(\boldsymbol{K_t})|
+
|\hat C^\star-C^\star|)
\le
2P_{\max}\bar\delta_{\mathrm{noise}}$.
Combined with (\ref{comp},\ref{ceregretori}), we obtain (\ref{eq:regret-final}).

\section{Proof of Theorem \ref{thm:adam-regret}}\label{app:thm:adam-regret}
The proof follows the same descent-based argument as the GD-DeePO regret analysis in Theorem \ref{thm:SGD}, with the standard gradient step replaced by the ADAM-preconditioned update. By the local smoothness property of $J_t$ in DeePO\cite{zhao2025data}, the one-step change of $J_t$ along the ADAM update direction satisfies
\begin{equation}\label{eq:adam-smooth-step}
\begin{aligned}
J_t(\boldsymbol{V}_{t+1})-J_t(\boldsymbol{V}_t)
\le
&-\tilde{\eta}_t\big\langle \nabla J_t(\boldsymbol{V}_t), \boldsymbol{D}_t\widehat{\boldsymbol m}_t \big\rangle\\
&+
\frac{l(J_t(\boldsymbol{V}_t))}{2}\tilde{\eta}_t^2
\|\boldsymbol{D}_t\widehat{\boldsymbol m}_t\|_F^2.
\end{aligned}
\end{equation}

Under Assumptions~\ref{ass:adam-pre}--\ref{ass:adam-steps}, Lemma~22 in~\cite{zou2019sufficient} implies that the preconditioned ADAM direction remains sufficiently aligned with the true gradient and has bounded magnitude relative to it. Hence, there exist constants $h_1,h_2>0$ such that
\begin{equation}\label{eq:adam-inner-bounds}
\begin{aligned}
&\big\langle \nabla J_t(\boldsymbol{V}_t), \boldsymbol{D}_t\widehat{\boldsymbol m}_t \big\rangle
\ge h_1 \|\nabla J_t(\boldsymbol{V}_t)\|_F^2,\\
&\|\boldsymbol{D}_t\widehat{\boldsymbol m}_t\|_F^2
\le h_2 \|\nabla J_t(\boldsymbol{V}_t)\|_F^2.
\end{aligned}
\end{equation}
Substituting~\eqref{eq:adam-inner-bounds} into~\eqref{eq:adam-smooth-step} gives
\begin{equation}\label{ineqadam}
J_t(\boldsymbol{V}_{t+1})-J_t(\boldsymbol{V}_t)
\le
-\tilde{\eta}_t
\Big(h_1-\frac{l(J_t(\boldsymbol{V}_t))}{2}h_2\tilde{\eta}_t\Big)
\|\nabla J_t(\boldsymbol{V}_t)\|_F^2.
\end{equation}
By absorbing $h_1$ and $h_2$ into the projected gradient-dominance modulus and smoothness constant, \eqref{ineqadam} yields the ADAM analogue of the one-step descent inequality used in the GD-DeePO analysis. The remainder of the regret proof then follows exactly as in the proof of Theorem~\ref{thm:SGD}, together with the covariance-mismatch term $2P_{\max}\bar{\delta}_{\mathrm{noise}}$. This gives~\eqref{eq:regret-adam}.

\section{Additional mechanism validation on a three-dimensional benchmark}
\label{app:3dbenchmark}

This appendix reports an additional low-dimensional benchmark study to illustrate two supporting mechanisms of the proposed method, namely disturbance-covariance estimation and adaptive gradient updates. It is not intended to represent the scale of industrial DHSs, but rather to provide a simple setting in which these effects can be isolated. We consider the marginally unstable Laplacian system from~\cite{dean2020sample,zhao2025data}, with $\boldsymbol{A}=\begin{bmatrix}1.01 & 0.01 & 0 \\ 0.01 & 1.01 & 0.01 \\ 0 & 0.01 & 1.01\end{bmatrix}$. The state and control weighting matrices are $\boldsymbol{Q}=\boldsymbol{R}=\boldsymbol{I}_3$. To evaluate DeePO under stochastic disturbances, we consider Gaussian process noise $\boldsymbol{w}_t \sim \mathcal{N}(0,\tfrac{1}{100}\boldsymbol{I}_3)$ and apply a PE input $\boldsymbol{u}_t=\boldsymbol{K}_t\boldsymbol{x}_t+\boldsymbol{v}_t$, where $\boldsymbol{v}_t \sim \mathcal{N}(0,\boldsymbol{I}_3)$, yielding $\mathrm{SNR}\in[0,5]$ as in~\cite{zhao2025data}. The DeePO algorithm is initialized using an LQR controller designed for a perturbed model $\boldsymbol{A}_{\mathrm{design}}=(1+\delta_e)\boldsymbol{A}$, where $\delta_e$ denotes the percentage model mismatch, and a 50-step warm-start dataset is collected.

\subsection{Effect of disturbance-covariance estimation}

To assess the effect of disturbance-covariance estimation, we test four input matrices with increasing coupling levels: decoupled $\boldsymbol{B}_{w1}=\begin{bmatrix}1&0&0\\0&1&0\\0&0&1\end{bmatrix}$, moderately coupled $\boldsymbol{B}_{w2}=\begin{bmatrix}1&1&0\\0&1&0\\0&0&1\end{bmatrix}$, strongly coupled $\boldsymbol{B}_{w3}=\begin{bmatrix}1&1&0.5\\0&1&0\\0&0&1\end{bmatrix}$, and highly coupled $\boldsymbol{B}_{w4}=\begin{bmatrix}1&1&1\\0&1&0\\0&0&1\end{bmatrix}$. A $10\%$ model perturbation is used to initialize a stabilizing gain $\boldsymbol{K}$.
We compare DeePO with $\boldsymbol{U}_\epsilon=\boldsymbol{I}$ against DeePO using the estimated disturbance covariance $\hat{\boldsymbol{U}}_\epsilon$ across all four input matrices, as shown in Figure~\ref{fig:3destimationdeepo}. Although the optimal LQR gain $\boldsymbol{K}^\star$ is theoretically independent of $\boldsymbol{U}_\epsilon$, the surrogate gradient used by DeePO depends on the closed-loop covariance. Consequently, assuming $\boldsymbol{U}_\epsilon=\boldsymbol{I}$ distorts the effective cost landscape when disturbances are anisotropic. In contrast, using $\hat{\boldsymbol{U}}_\epsilon$ provides more accurate gradient directions and leads to smoother and more stable convergence. This advantage becomes increasingly pronounced from $\boldsymbol{B}_{w1}$ to $\boldsymbol{B}_{w4}$ as the disturbance coupling strength increases.

\vspace{-0.3cm}
\begin{center}
  \includegraphics[width=\linewidth]{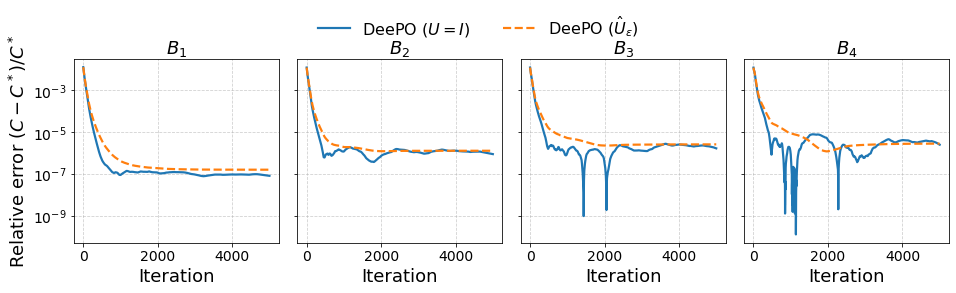}
  \vspace{-0.6cm}
  \captionof{figure}{Effect of disturbance-covariance estimation.}
  \label{fig:3destimationdeepo}
\end{center}
\vspace{-0.5cm}

\subsection{Comparison with zeroth-order policy optimization (ZO-PO)}

We further compare DeePO with classical zeroth-order policy optimization (ZO-PO) on the strongly coupled case $\boldsymbol{B}_{w3}$, using the same initial stabilizing controller and disturbance level ($\sigma_w=0.01$). ZO-PO estimates gradients by two-point finite differences, whereas DeePO updates the policy directly from closed-loop covariance data. Figure~\ref{fig:3destimationzero}(a) shows that DeePO converges faster and more smoothly, while ZO-PO is slower and more oscillatory due to noisy gradient estimates. Figure~\ref{fig:3destimationzero}(b) further shows that ZO-PO requires about 3000 samples per update, whereas DeePO uses only one, giving DeePO much higher sample efficiency.

\vspace{-0.2cm}
\begin{center}
    \centering
    \includegraphics[width=1.0\linewidth]{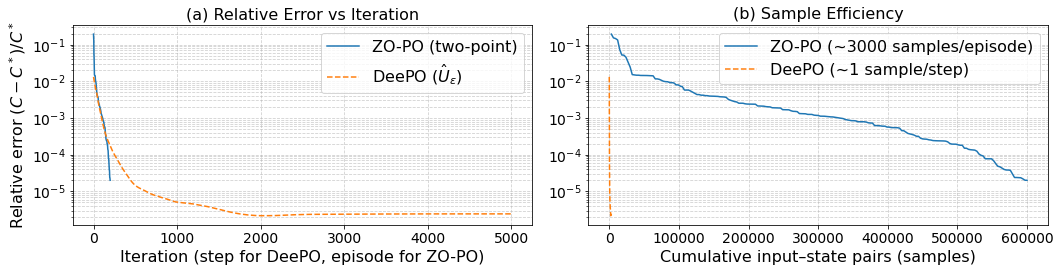}
    \vspace{-0.5cm}
    \captionof{figure}{Sample efficiency of DeePO and ZO-PO.}
    \label{fig:3destimationzero}
\end{center}
\vspace{-0.8cm}

\section*{Data availability}
The data and code that support the findings of this study are available from the corresponding author upon reasonable request.

\bibliographystyle{cas-model2-names}
\bibliography{cas-refs}

\end{document}